\documentclass[11pt]{article}
\usepackage{graphicx,lscape}
\usepackage{amsmath,amsthm,amssymb}
\usepackage{amsfonts,mathrsfs}
\usepackage{bm,indentfirst}

\begin{document}
\renewcommand{\d}{\mathrm{d}}
\newtheorem{thm}{Theorem}[section]
\newtheorem{lemma}[thm]{Lemma}
\newtheorem{propos}[thm]{Proposition}
\newtheorem{corol}[thm]{Corollary}
\numberwithin{equation}{section}
\newcommand{\expct}{\mbox{E}}
\renewcommand{\H}{\mathscr{H}}
\newcommand{\F}{\mathcal{F}}
\newcommand{\M}{\mathcal{M}}

\title{From microscopic theory to macroscopic theory: symmetries and order parameters of rigid molecules}
\author{Jie Xu$^{*}$ and Pingwen Zhang$^{*}$\\[2mm]
{\small LMAM \& School of  Mathematical Sciences, Peking University, China}\\
{\small $^*$ E-mail: rxj\_2004@126.com,\, pzhang@pku.edu.cn}\\
}
\date{\today}
\maketitle

\begin{abstract}
We use density functional theory to describe the phase behaviors of rigid molecules.
The construction of kernel function $G(\bm{x},P,\bm{x'},P')$ is discussed. 
Excluded-volume potential is calculated for two types of molecules with 
$C_{2v}$ symmetry. 
Molecular symmetries lead to the symmetries of $G$ and density function 
$f(P)$, enabling a reduction of configuration space. 
By approximating $G$ with a polynomial, 
the system can be fully characterized by some moments corresponding to the form of $G$. 
The symmetries of $G$ determine the form of the polynomial, 
while the coefficients are determined by temperature and molecular parameters. 
The analysis of the impact of coefficients helps us to choose 
independent variables in the moments as order parameters. 
Order parameters for bent-core molecules are predicted. 
\end{abstract}
\tableofcontents
\setlength\arraycolsep{2pt}
\section{Introduction}
Non-spherical rigid molecules may show orientationally ordered phases, 
which are often referred to as liquid crystalline(LC) phases. 
Among all rigid molecules, (nonpolar) rod-like molecules have attracted most 
interests. Rod-like molecules have shown several LC phases experimentally, 
including nematic, smectic A and smectic C. The shape of rigid molecules can be 
more complex, which may induce richer phase behaviors. In the recent two 
decades, a novel type of molecules has occupied a place in the 
research of liquid crystals. This type of molecules can be represented by two 
rigid rods connected end to end with a fixed angle, thus called bent-core 
molecules. Bent-core molecules break the rotational symmetry of rod-like 
molecules, and have proven to exhibit numerous new liquid crystalline phases\cite{JJAP}. 
A few rigid molecules of other architectures have also been studied and the 
experimental results indicate more complex phases\cite{NT}. 

To describe LC phases, orientation-dependent variables are necessary to be 
included in the free energy. 
The statistical mechanics of nematic phase of rods gives the free energy
\begin{equation}\label{FreeEng0}
  F[f]=F_0+k_BT\left(\int_{S^2}\d\bm{m} f(\bm{m})\log f(\bm{m}) +
\frac{c}{2} \int\int_{S^2\times S^2} \d\bm{m}\d\bm{m'}f(\bm{m})G(\bm{m},\bm{m'})f(\bm{m'})
\right)
\end{equation}
where $f(\bm{m})$ is orientational probability density function, 
$G(\bm{m},\bm{m'})$ is kernel function, and $c$ is concentration. 
Two of the most well-known $G$ are Onsager potential\cite{Ons} 
\begin{equation}\label{Onsager}
2L^2D|\bm{m}\times\bm{m'}|
\end{equation}
for a rod of length $L$ and diameter $D$, and Maier-Saupe potential\cite{M_S}
\begin{equation}\label{MS}
C|\bm{m}\times\bm{m'}|^2=C-C(\bm{m}\cdot\bm{m'})^2
\end{equation}
with $C$ related to $L$, $D$ and temperature $T$. 
The top eigenvalue of the second moment of $\bm{m}$, 
\begin{equation}
S=\left<\bm{mm}\right>, 
\end{equation}
is defined as the order parameter of the uniaxial nematic phase\cite{LiqCryst}, 
which is the only spatially homogeneous phase found for rod-like molecules 
other than isotropic phase. 

The free energy (\ref{FreeEng0}) is a natural extension of virial expansion of spheres. 
When handling a generic rigid molecule, a three-dimensional rotation 
$P\in SO_3$ is necessary to describe its orientation. Thus we need to 
substitute $\bm{m}$ with $P$ in (\ref{FreeEng0}). Kernel function $G(P,P')$ 
can be deduced from pairwise interaction of molecules. 
Different phases correspond to different local minima of the free energy. 
However, it is obscure to distinguish phases with the probability density 
function $f$. One always wants to seek a few order parameters to classify them, 
like the eigenvalues of $\left<\bm{mm}\right>$ for nematic phase of rods. 

In the existing approaches, order parameters are usually considered at first. 
Models at different levels are constructed about these order parameters. 
For example, in \cite{OrdPar, BiLand} Landau-type free energies are constructed 
for molecules with $C_{2v}$ and $D_{2h}$ symmetries. 
In \cite{Bi1} a molecular theory of $D_{2h}$ is dicussed. Four order parameters 
are proposed. The kernel function there is a polynomial of 
$\bm{m}_i\cdot\bm{m'}_j$. When solving the model, further assumptions are made 
to deduce the equations of the four order parameters. 

The purpose of this article is to present a procedure of reducing $f$ to 
a few order parameters. 
In this procedure, symmetries of molecules play a key role. These symmetries 
will be inherited by kernel function $G$ and probability density function $f$. 
The symmetries of $G$ and $f$ make it possible to reduce the configuration 
space. As an example, the reduction will derive (\ref{FreeEng0}) for molecules 
with axial symmetry. 

The next step is to look for a good approximation of $G$. 
Here we are partially inspired by some thoughts in \cite{Bi1}. 
We will prove that 
$G$ is a function of $\bar{P}=P^{-1}P'=(p_{ij})_{3\times 3}$, and approximate $G$
with a polynomial of $p_{ij}$. The advantage of polynomial approximation is 
that the Euler-Langrange equation of $f$ could be replaced with self-consistent 
equations of several moments of $\bm{m}_i$ that fully characterize the system. 
The symmetries of $G$ determine the form of approximate kernel function. 
In other words, the symmetries of $G$ determine the candidate moments. 
Truncation within the remaining terms is followed, which 
relies on intuitions from experiments and simulations. Maier-Saupe potential 
is obtained spontaneously after this step for molecules with $D_{\infty h}$ 
symmetry, and the form of approximation is derived for molecules with $C_{2v}$ 
symmetry. 

The coefficients of polynomial approximation of $G$ are determined by
molecular parameters and temperature. The analysis of the impacts of these 
coefficients might further reveal some properties of the chosen moments. 
Analysis of this type has been done for rods\cite{AxiSymMS,AxiSym2,AxiSym3} and polar rods\cite{Dipol}. 
We will present some results on molecules with $C_{2v}$. 
The analysis would enable us to find independent 
variables in these moments, which are chosen as order parameters. 
From the implications of the analysis and the experimental results, 
we predict that five order parameters are enough for bent-core molecules. 

The rest of this paper is organized as follows. Sec.\ref{Model} describes 
the density functional theory of generic rigid molecules. 
The construction of kernel function $G$ is discussed. Some simple properties
are presented. Excluded-volume potential is derived for two types of 
molecules with $C_{2v}$ symmetry. 
In Sec.\ref{Sy}, we analyze the symmetric properties of $G$ and $f$ and 
describe the reduction of configuration space. 
Sec.\ref{Trunc} shows the derivation of the equations of moments and 
the screening of the moments by symmetries of kernel function. 
Sec.\ref{Ord} is dedicated to the analysis of the impacts of the coefficients 
in polynomial approximation of $G$. 
In Sec.\ref{Con}, we make a conclusion and propose some prospective problems. 

\section{Modelling of rigid molecules of arbitrary shape}\label{Model}
This section presents the density functional theory of rigid molecules. 
We start from a general formulation, then deduce the free energy for spatially 
homogeneous phases. 
A three-dimensional rigid molecule might be chiral, leading to two possible 
configurations that cannot coincide through proper rotation. 
In this work we simply deal with systems with single chirality. 
Systems with mixed chirality can be treated by regarding two kinds of chirality 
as different molecules. 
\subsection{Representation of the configuration of rigid molecules}
We choose a reference point $\hat{O}$ on the rigid molecule and a body-fixed orthogonal 
basis $\bm{m}_1,\bm{m}_2,\bm{m}_3$. The configuration of the molecule is 
determined by the position of $\hat{O}$ and the orientation of $\bm{m}_i$. In a 
space-fixed orthogonal coordinate system $(O;\bm{e}_1,\bm{e}_2,\bm{e}_3)$, they 
can be expressed in terms of $\bm{x}_0=\overrightarrow{O\hat{O}}$ 
and a three-dimensional proper rotation $P\in SO_3$. 
In the language of matrix, $P$ is orthogonal such that
\begin{equation}\label{RotP}
\left(
  \bm{m}_1,
  \bm{m}_2,
  \bm{m}_3
\right)
=
\left(
  \bm{e}_1,
  \bm{e}_2,
  \bm{e}_3
\right)P.
\end{equation}
The elements of $P^T=(m_{ij})$ is given by
$$
m_{ij}=\bm{m}_i\cdot\bm{e}_j.
$$
The position of a fixed point on the molecule is represented by its coordinates 
in the body-fixed coordinate system $(O;\bm{m}_1,\bm{m}_2,\bm{m}_3)$: 
$$
\bm{\hat{x}}=\left(
\begin{array}{c}
\hat{x}_1\\
\hat{x}_2\\
\hat{x}_3
\end{array}
\right).
$$
Its location in space, expressed by its coordinates in $(O;\bm{e}_1,\bm{e}_2,\bm{e}_3)$, is
$$
\bm{x}=\left(
\begin{array}{c}
x_1\\
x_2\\
x_3
\end{array}
\right)
=P\bm{\hat{x}}+\bm{x}_0.
$$

Every $P\in SO_3$ has a representation by Euler angles $\alpha, \beta, \gamma$: 
\begin{eqnarray}
&&P(\alpha,\beta,\gamma)\nonumber\\
&=&\left(
\begin{array}{ccc}
 \cos\alpha &\quad -\sin\alpha\cos\gamma &\quad\sin\alpha\sin\gamma\\
 \sin\alpha\cos\beta &\quad\cos\alpha\cos\beta\cos\gamma-\sin\beta\sin\gamma &
 \quad -\cos\alpha\cos\beta\sin\gamma-\sin\beta\cos\gamma\\
 \sin\alpha\sin\beta &\quad\cos\alpha\sin\beta\cos\gamma+\cos\beta\sin\gamma &
 \quad -\cos\alpha\sin\beta\sin\gamma+\cos\beta\cos\gamma
\end{array}
\right),\label{EulerRep}
\end{eqnarray}
with
$$
\alpha\in [0,\pi],\ \beta,\gamma\in [0,2\pi). 
$$
The uniform probability measure on $SO_3$ is given by
$$
\d\nu=\frac{1}{8\pi^2}\sin\alpha\d\alpha\d\beta\d\gamma. 
$$

\subsection{Density functional theory}
We start from the extension of virial expansion that includes inhomogeneity 
both spatial and orientational. 
\begin{equation}\label{FreeEng1}
\begin{split}
F[f]&=F_0+\frac{k_BT}{V}
\left[\int \d\nu\d\bm{x} f(\bm{x},P)\log f(\bm{x},P)\right. \\
&\left.+\frac{1}{2}\int\d\nu(P)\d\bm{x}\d\nu(P')\d\bm{x'}
f(\bm{x},P)G(\bm{x},P,\bm{x'},P')f(\bm{x'},P')\right]. 
\end{split}
\end{equation}
The probability density function $f$ agrees with the concentration $c$: 
$$
\frac{1}{V}\int\d\bm{x}\int\d\nu f(\bm{x},P)=c. 
$$
Virial expansion is appropriate for small concentration. 
Corrections for large concentration have also been discussed, such as 
in \cite{Largec}. 
The kernel function in (\ref{FreeEng1}) is given by Mayer function\cite{Mayer}
\begin{equation}\label{VirExp}
  G(\bm{x},P,\bm{x'},P')=1-\exp\left(-U(\bm{x},P,\bm{x'},P')/k_BT\right)
\end{equation}
where $U$ is pairwise interaction. 

Many types of interaction could appear in $U$. But here we model the molecule 
as a combination of spheres with the same diameter $D$ and assume that $U$ 
consists of the sum of interaction of every pairs of spheres. 
Suppose that the distribution of their centers is given by $\rho(\bm{\hat{x}})$ 
in the body-fixed coordinate system, then $U$ can be written as 
\begin{equation}\label{Interaction}
  U(\bm{x},P,\bm{x'},P')=\int\d\bm{\hat{x}}\d\bm{\hat{x}'}
  V_0\Big(|(P\bm{\hat{x}}+\bm{x})-(P'\bm{\hat{x}'}+\bm{x'})|\Big)
  \rho(\bm{\hat{x}})\rho(\bm{\hat{x}'})
\end{equation}
where $V_0(r)$ is the potential of a single pair of spheres. 
It can take hardcore potential
\begin{equation}\label{hardcore}
  V_0(r)=\left\{
  \begin{array}{cc}
    \infty,&r\le D\\
    0,&r>D
  \end{array}
  \right.
  ,
\end{equation}
or Lennard-Jones potential
\begin{equation}\label{LJ}
  V_0(r)=4\epsilon\left[
  \left(\frac{\sigma}{r}\right)^{12}-\left(\frac{\sigma}{r}\right)^6\right]. 
\end{equation}
In (\ref{LJ}) $\sigma$ is a function of $D$. Some other types of interaction between spheres 
can also be incoporated in $V_0$, such as electrostatic potential for charged 
molecules. 
Independent of $V_0$, kernel function $G$ has the following properties: 
\begin{propos}\label{Inv0}
  \begin{enumerate}
  \item 
  $G(\bm{x},P,\bm{x'},P')$ remains unchanged when switching $(\bm{x},P)$ and 
  $(\bm{x'}, P')$: 
  \begin{equation}\label{InvSwap}
  G(\bm{x},P,\bm{x'},P')=G(\bm{x'},P',\bm{x},P). 
  \end{equation}
  \item 
  $G(\bm{x},P,\bm{x'},P')$ depends only on $\bm{x-x'}$ when $P$ and $P'$ 
  are fixed: 
  \begin{equation}\label{Inv1}
  G(\bm{x},P,\bm{x'},P')=G(\bm{x-x'},P,P').
  \end{equation}
  \item 
  $G(\bm{x},P,\bm{x'},P')$ is invariant when two molecules rotate together: 
  \begin{equation}\label{Inv2}
  G\big(T(\bm{x-x'}),TP,TP'\big)=G(\bm{x-x'},P,P'),\ \forall T\in SO_3. 
  \end{equation}
  \end{enumerate}
\end{propos}
\begin{proof}
  From (\ref{VirExp}) it is sufficient to show 
  \begin{eqnarray*}
  &&U(\bm{x},P,\bm{x'},P')=U(\bm{x'},P',\bm{x},P)=U(\bm{x-x'},P,P'),\\
  &&U\big(T(\bm{x-x'}),TP,TP'\big)=U(\bm{x-x'},P,P').
  \end{eqnarray*}
  The former is obvious from the definition (\ref{Interaction}) of $U$. For the 
  latter, 
  \begin{eqnarray*}
  U\big(T(\bm{x-x'}),TP,TP'\big)&=&\int\d\bm{\hat{x}}\d\bm{\hat{x}'}
  V_0\Big(\left|(TP\bm{\hat{x}}+T\bm{x})-(TP'\bm{\hat{x}'}+T\bm{x'})\right|\Big)
  \rho(\bm{\hat{x}})\rho(\bm{\hat{x}'})\\
  &=&\int\d\bm{\hat{x}}\d\bm{\hat{x}'}
  V_0\Big(|T|\cdot\left|(P\bm{\hat{x}}+\bm{x})-(P'\bm{\hat{x}'}+\bm{x'})\right|\Big)
  \rho(\bm{\hat{x}})\rho(\bm{\hat{x}'})\\
  &=&U(\bm{x-x'},P,P').
  \end{eqnarray*}
\end{proof}

Now we deduce the free energy of spatially homogeneous phases, namely
$$
f(\bm{x},P)=c\tilde{f}(P)
$$
where $\tilde{f}(P)$ is a density function on $SO_3$. Define homogeneous kernel 
function as
\begin{equation}\label{HomG}
\tilde{G}(P,P')=\int\d\bm{x'} G(\bm{x-x'},P,P'). 
\end{equation}
It is well-defined because the integration on the right side is 
invariant with $\bm{x}$: 
\begin{eqnarray*}
  \int\d\bm{x'} G(\bm{x-x'},P,P')&=&\int\d(\bm{x'-x}) G(\bm{x-x'},P,P')\\
  &=&\int\d\bm{x'} G(\bm{0-x'},P,P'),\qquad \forall \bm{x}\in \mathbb{R}^3. 
\end{eqnarray*}
Applying Proposition \ref{Inv0} to $\tilde{G}$, we have
\begin{propos}\label{RelP}
  $\tilde{G}(P,P')$ satisfies
  $$
  \tilde{G}(P,P')=\tilde{G}(P',P)
  $$
  and
  $$
  \tilde{G}(P,P')=\tilde{G}(TP,TP'). 
  $$
  By setting $T=P^{-1}$, we know that $\tilde{G}(P,P')$ is a function of 
  $\bar{P}=P^{-1}P'$, which is denoted by $\tilde{G}(\bar{P})$. 
\end{propos}
\begin{proof}
  Using (\ref{InvSwap}) and (\ref{Inv2}), we get
  \begin{eqnarray*}
    \tilde{G}(P',P)&=&\int\d\bm{x'}G(\bm{x-x'},P',P)\\
    &=&\int\d\bm{x'}G(\bm{x'-x},P,P')=\int\d\bm{x'}G(\bm{-x-x'},P,P')\\
    &=&\tilde{G}(P,P'), \\\\
    \tilde{G}(TP,TP')&=&\int\d\bm{x'}G\big(T(\bm{x-x'}),TP,TP'\big)\\
    &=&\int\d\bm{x'}G(\bm{x-x'},P,P')=\tilde{G}(P,P'). 
  \end{eqnarray*}
\end{proof}
The rest of paper will focus on spatially homogeneous phases. 
For convenience we use $f(P)$ and $G(P,P')$ instead of 
$\tilde{f}(P)$ and $\tilde{G}(P,P')$. 
Then the free energy (\ref{FreeEng1}) becomes
\begin{eqnarray}
  \frac{F[f]}{c}&=&\frac{F_0}{c}+k_BT\log c+k_BT
  \left[\int \d\nu f(P)\log f(P)\right. \nonumber\\
  &&\left.+\frac{c}{2}\int\d\nu(P)\d\nu(P') 
    f(P)G(P,P')f(P')\right]\label{FreeEngN}
\end{eqnarray}
with the normalization condition of $f(P)$
\begin{equation}\label{Consrv}
\int\d\nu f(P)=1. 
\end{equation}

For rod-like and bent-core molecules, the sphere centers lie on a curve. 
They can be viewed as either discretely or continuously distributed on the 
curve, which means 
$$
\rho(\bm{\hat{x}})=\sum_{j=0}^N\delta(\hat{\bm{x}}-\bm{\tilde{r}}_j), 
$$
or 
$$
\rho(\bm{\hat{x}})=\int\d s \delta\big(\bm{\hat{x}}-\bm{\tilde{r}}(s)\big). 
$$
In the discrete version, a rod-like molecule is modelled by
\begin{equation}\label{DisR}
  \bm{\tilde{r}}_j=Ls_j\bm{m}_1,\ s_j=\frac{j}{N}-\frac{1}{2},\ L=(N-1)r_0, 
\end{equation}
and a bent-core molecule is modelled by 
\begin{equation}\label{DisB}
  \bm{\tilde{r}}_j=L(\frac{1}{2}-|s_j|)\cos\frac{\theta}{2}\bm{m}_1  +Ls_j
  \sin\frac{\theta}{2}\bm{m}_2,\ s_j=\frac{j}{N}-\frac{1}{2},\ L=(N-1)r_0,  
\end{equation}
where $N$ is even. In the continum version, a rod-like molecule is modelled by
\begin{equation}\label{ContR}
  \bm{\tilde{r}}(s)=Ls\bm{m}_1,\ s\in[-\frac{1}{2},\frac{1}{2}],
\end{equation}
and a bent-core molecule is modelled by
\begin{equation}\label{ContB}
  \bm{\tilde{r}}(s)=L(\frac{1}{2}-|s|)\cos\frac{\theta}{2}\bm{m}_1
  +Ls\sin\frac{\theta}{2}\bm{m}_2,\ s\in[-\frac{1}{2},\frac{1}{2}].
\end{equation}
\begin{figure}
\centering
\includegraphics[width=0.87\textwidth,keepaspectratio]{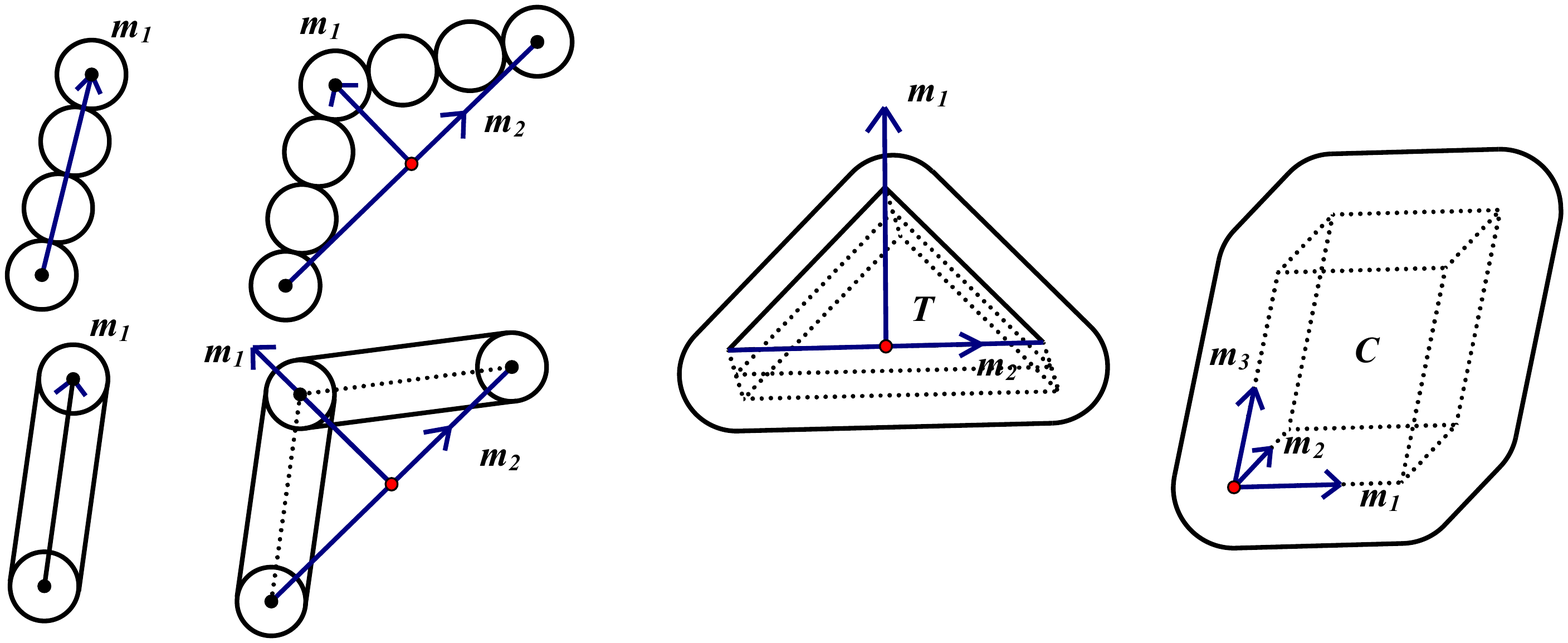}
\caption{Different rigid molecules. From left to right: rod, 
bent-core molecule, spherotriangle, spherocuboid.}\label{molecules}
\end{figure}
Both the discrete and continuous versions have the same symmetry: rod-like 
molecules have $D_{\infty h}$ symmetry, and bent-core molecules have $C_{2v}$ 
symmetry. Another example of $C_{2v}$ symmetry is isosceles spherotriangles 
(by sphero-A, we refer to the {\it Minkowski sum} of A and a sphere): 
\begin{equation}
  \rho(\bm{\hat{x}})=\int_T \d S \delta(\bm{\hat{x}}-\bm{r})
\end{equation}
where $T$ is an isosceles triangle that lies in plane $O\bm{m}_1\bm{m}_2$. 
Spherocuboids, which possess $D_{2h}$ symmetry, is considered in \cite{Bi1}. 
The distribution of sphere centers is given by 
\begin{equation}
  \rho(\bm{\hat{x}})=\int_C \d V \delta(\bm{\hat{x}}-\bm{r})=\chi_{\{\bm{\hat{x}}\in C\}}
\end{equation}
where $C$ is a cuboid with edges parallel to $\bm{m}_i$. 
The molecules mentioned above are drawn in Fig.\ref{molecules}. 

Next we try to compute $G(P,P')$ when using hardcore potential 
(\ref{hardcore}). Notice that $G(\bm{x},P,\bm{x'},P')$ equals to $1$ if two 
molecules overlap and $0$ elsewhere. Thus by (\ref{HomG}) $G(P,P')$ is actually 
excluded-volume potential, and the problem has converted into finding this 
volume. Onsager potential (\ref{Onsager}) is an approximation of the excluded 
volume for rod-like molecules. Similar approach has been done for 
spherozonotopes by B. M. Mulder\cite{Mulder} based on {\it Steiner formula}
\cite{Convex}. Here we present the calculation of the excluded volume of 
spherotriangles.

Denote the set of sphere centers of two molecules as $T_1$ and $T_2$. 
The excluded region of two molecules is $K+B_D$, where $K=T_1-T_2$ and 
$B_D$ is a sphere of diameter D. Here $A-B$ represents
$$
A-B=\{a-b|a\in A, b\in B\}. 
$$
When $K$ is convex, 
the measure of $K+B_D$ is expressed by Steiner formula in a polynomial of $D$.
Here we write down the three-dimensional case 
\begin{equation}\label{Steiner}
  V(K+B_D)=V_3(K) + DV_2(K) + \pi D^2V_1(K) + \frac{4}{3}\pi D^3, 
\end{equation}
where $V_3$ is the volume and $V_2$ is the surface area of $K$(see p.210 of 
\cite{Convex}). $V_1$ is the {\it mean width} of $K$ (for the definition, 
see p.42 of \cite{Convex}). For a polytope, $V_1$ is written as 
\begin{equation}\label{ExtRep}
  V_1=\sum_{e}\gamma(e,K)l(e).
\end{equation}
The sum is taken over all the edges of $K$. $l(e)$ is the length of edge, and 
$\gamma(e,K)$ represents the external angle at $e$. It is defined as 
$$
\gamma(e,K)=\frac{\theta}{2\pi}, 
$$
where $\theta$ is the angle between outward normal vectors of two faces that 
share edge $e$. 
For a polytope $K$, each of the four terms in (\ref{Steiner}) represents a part 
of $K+B_D$: $K$; parallelepipeds growing outward at each face; circular sector 
cylinders at each edge; corner regions at each vertex. 

We list $V_i$ for the excluded volume of two spherotriangles 
$T_1=\triangle OAB$ and $-T_2=\triangle O'A'B'$. 
Details are left to Appendix. Denote the edges of triangles as 
$$
\overrightarrow{AO}=\bm{a},\ \overrightarrow{OB}=\bm{b},\ \overrightarrow{BA}=\bm{c},\ 
\overrightarrow{A'O'}=\bm{a'},\ \overrightarrow{O'B'}=\bm{b'},\ \overrightarrow{B'A'}=\bm{c'}.
$$
We have
\begin{align*}
  &V_3(K)=
  \frac{1}{2}\Big(|\bm{a}\times\bm{b}\cdot\bm{a'}|+|\bm{a}\times\bm{b}\cdot\bm{b'}|
    +|\bm{a}\times\bm{b}\cdot\bm{c'}|+|\bm{a'}\times\bm{b'}\cdot\bm{a}|
    +|\bm{a'}\times\bm{b'}\cdot\bm{b}|+|\bm{a'}\times\bm{b'}\cdot\bm{c}|\Big).\\
  &V_1(K)=\frac{1}{2}\Big(|\bm{a}|+|\bm{b}|+|\bm{c}|+|\bm{a'}|+|\bm{b'}|+|\bm{c'}|\Big).
\end{align*}
The expression of $V_2$ depends on the relative orientation of two molecules. 
Assume that 
\begin{eqnarray*}
  (\bm{a}\times\bm{b}\cdot\bm{a'})(\bm{a}\times\bm{b}\cdot\bm{b'})\ge 0, \\
  (\bm{a'}\times\bm{b'}\cdot\bm{a})(\bm{a'}\times\bm{b'}\cdot\bm{b})\ge 0. 
\end{eqnarray*}
Otherwise we could rotate the notations of the edges. 
If $(\bm{c}\times\bm{c'}\cdot\bm{a})(\bm{c}\times\bm{c'}\cdot\bm{a'})>0$, then 
$$
V_2(K)=|\bm{a}\times \bm{b}|+|\bm{a'}\times \bm{b'}|
+|\bm{a}\times \bm{a'}|+|\bm{a}\times \bm{b'}|+|\bm{b}\times \bm{a'}|
+|\bm{b}\times \bm{b'}|+|\bm{c}\times \bm{c'}|.
$$
If $(\bm{c}\times\bm{c'}\cdot\bm{a})(\bm{c}\times\bm{c'}\cdot\bm{a'})<0$, then 
$$
V_2(K)=|\bm{a}\times \bm{b}|+|\bm{a'}\times \bm{b'}|
+|\bm{c}\times \bm{a'}|+|\bm{c}\times \bm{b'}|
+|\bm{a}\times \bm{c'}|+|\bm{b}\times \bm{c'}|.
$$

When $K$ is not convex, $V(K+B_D)$ does not have a general formula. In Appendix 
we give an expression for the excluded volume of bent-core molecules. 

Excluded-volume potential fails to contain temperature $T$ in 
kernel function, which makes it insufficient to study thermotropic LC. 
In the next two sections we will present a systematic procedure of constructing 
an approximation of kernel function that does not eliminate temperature. 
Symmetries play an important role throughout the procedure.

\section{Symmetries of kernel function and reduction of configuration space}
\label{Sy}
In this section we study the symmetric properties of $G$ and $f$ inherited from 
molecular symmetries. That a rigid molecule is symmetric under $T\in SO_3$ means
\begin{equation}\label{Symm0} 
\rho(T\bm{\hat{x}})=\rho(\bm{\hat{x}}). 
\end{equation}
Denote by $\H$ a subgroup of $SO_3$ that leaves the molecule invariant. 
We start from two fundamental theorems. 
\begin{thm}
\label{TInv}
If $T\in\H$, then
\begin{equation}\label{SymmG}
G(PT,P')=G(P,P'T)=G(P,P').
\end{equation}
\end{thm}
\proof
From (\ref{Symm0}) and (\ref{Interaction}), $U$ is symmetric under $T$: 
\begin{eqnarray*}
  U(\bm{x},PT,\bm{x'},P')&=&\int\d\bm{\hat{x}}\d\bm{\hat{x}'}
  V_0\Big(\left|(PT\bm{\hat{x}}+\bm{x})-(P'\bm{\hat{x}'}+\bm{x'})\right|\Big)
  \rho(\bm{\hat{x}})\rho(\bm{\hat{x}'})\\
  &=&\int\d(T\bm{\hat{x}})\d\bm{\hat{x}'}
  V_0\Big(\left|(P(T\bm{\hat{x}})+\bm{x})-(P'\bm{\hat{x}'}+\bm{x'})\right|\Big)
  \rho(T\bm{\hat{x}})\rho(\bm{\hat{x}'})\\
  &=&\int\d\bm{\hat{x}}\d\bm{\hat{x}'}
  V_0\Big(\left|(P\bm{\hat{x}}+\bm{x})-(P'\bm{\hat{x}'}+\bm{x'})\right|\Big)
  \rho(\bm{\hat{x}})\rho(\bm{\hat{x}'})\\
  &=&U(\bm{x},P,\bm{x'},P'). 
\end{eqnarray*}
By (\ref{VirExp}), we have
\begin{eqnarray*}
G(\bm{x-x'},P,P')&=&1-\exp\big(-U(\bm{x},P,\bm{x'},P')\big)\\
&=&1-\exp\big(-U(\bm{x},PT,\bm{x'},P')\big)\\
&=&G(\bm{x-x'},PT,P').  
\end{eqnarray*}
Hence by (\ref{HomG}), we get
\begin{eqnarray*}
G(P,P')&=&\int\d\bm{x'} G(\bm{x-x'},P,P')\\
&=&\int\d\bm{x'} G(\bm{x-x'},PT,P')=G(PT,P'). 
\end{eqnarray*}
The other equality in (\ref{SymmG}) could be obtained similarly. 
\qed

\begin{thm}
\label{SymP}
  If a molecule has a symmetry plane with unit normal vector $\bm{\hat{k}}$, 
  then
  \begin{equation}\label{SymPlane}
    G(J\bar{P}J)=G(\bar{P})
  \end{equation}
  where $J$ is the rotation around $\bm{\hat{k}}$ by $\pi$. 
\end{thm}
\begin{proof}
Assume that the symmetry plane contains $\hat{O}$, otherwise we shift the body-fixed 
coordinate system to meet this requirement. Now we have
  \begin{equation}\label{SymPlane0}
  \rho\big(\bm{\hat{x}}-2(\bm{\hat{k}}\cdot\bm{\hat{x}})\bm{\hat{k}}\big)=\rho(\bm{\hat{x}}).
  \end{equation}
Note that 
\begin{eqnarray*}
J\bm{\hat{x}}&=&J\big(\bm{\hat{x}}-(\bm{\hat{k}}\cdot\bm{\hat{x}})\bm{\hat{k}}\big)
+(\bm{\hat{k}}\cdot\bm{\hat{x}})J\bm{\hat{k}}\\
&=&-\big(\bm{\hat{x}}-(\bm{\hat{k}}\cdot\bm{\hat{x}})\bm{\hat{k}}\big)
+(\bm{\hat{k}}\cdot\bm{\hat{x}})\bm{\hat{k}}\\
&=&-\bm{\hat{x}}+2(\bm{\hat{k}}\cdot\bm{\hat{x}})\bm{\hat{k}}. 
\end{eqnarray*}
Substituting it into (\ref{SymPlane0}), we get
$$
\rho(-J\bm{\hat{x}})=\rho(\bm{\hat{x}}). 
$$
Therefore
\begin{eqnarray*}
  &&U(\bm{x},PJ,\bm{x'},P'J)\\
  &=&\int\d\bm{\hat{x}}\d\bm{\hat{x}'}
  V_0\Big(\left|(PJ\bm{\hat{x}}+\bm{x})-(P'J\bm{\hat{x}'})+\bm{x'}\right|\Big)
  \rho(\bm{\hat{x}})\rho(\bm{\hat{x}'})\\
  &=&\int\d(-J\bm{\hat{x}})\d(-J\bm{\hat{x}'})
  V_0\Big(\left|\big(P(-J\bm{\hat{x}})-\bm{x}\big)-\big(P'(-J\bm{\hat{x}'}\big)-\bm{x'})\right|\Big)
  \rho(-J\bm{\hat{x}})\rho(-J\bm{\hat{x}'})\\
  &=&\int\d\bm{\hat{x}}\d\bm{\hat{x}'}
  V_0\Big(\left|(P\bm{\hat{x}}-\bm{x})-(P'\bm{\hat{x}'}-\bm{x'})\right|\Big)
  \rho(\bm{\hat{x}})\rho(\bm{\hat{x}'})\\
  &=&U(-\bm{x},P,-\bm{x'},P'). 
\end{eqnarray*}
Similar to Theorem \ref{TInv}, we have
$$
G(\bm{x-x'},P,P')=G(\bm{x'-x},PJ,P'J).
$$
Thus
\begin{eqnarray*}
G(P,P')&=&\int\d\bm{x'} G(\bm{x-x'},P,P')\\
&=&\int\d\bm{x'} G(\bm{x'-x},PJ,P'J)=G(PJ,P'J). 
\end{eqnarray*}
\end{proof}
The local minima of the free energy (\ref{FreeEngN}) satisfy the Euler-Lagrange 
equation
\begin{equation}\label{EL}
  \lambda+\log f(P) + c\int\d\nu(P')G(P,P')f(P')=0
\end{equation}
where $\lambda$ is a Lanrangian multiplier to ensure the normalization of $f$. 
Denote 
\begin{equation}\label{DefU}
  W(P)=c\int\d\nu(P')G(P,P')f(P'). 
\end{equation}
The solution of (\ref{EL}) has the form
\begin{equation}\label{Boltz}
  f(P)=C\exp\big(-W(P)\big). 
\end{equation} 
\begin{thm}\label{fInv}
  If $T\in \H$, the solutions of (\ref{EL}) satisfy
  \begin{equation}\label{fInv1}
    f(P)=f(PT). 
  \end{equation}
\end{thm}
\proof
Substitute $P$ with $PT$ in (\ref{EL}). By Theorem \ref{TInv}, $G(PT,P')
=G(P,P')$, thus
$$
  \lambda+\log f(PT)=-\int\d\nu(P')G(PT,P')f(P')
    =-\int\d\nu(P')G(P,P')f(P')=\lambda+\log f(P).
$$
\qed

With the symmetries of $G$ and $f$, the configuration space could be 
reduced. Theorem.\ref{TInv} and \ref{fInv} indicate that $G$ is 
a function of $\H \bar{P}\H$, and $f$ is a function of $P\H$. Note that cosets 
of subgroup $\H$ form a partition of $SO_3$. Thus it allows us to define 
$\Omega=\{P\H | P\in SO_3\}$ as the new configuration space, where $f$ and $G$ 
are well-defined in $\Omega$ and $\Omega\times\Omega$, respectively. 
If we denote the probability space on $SO_3$ as $(SO_3, \F, \nu)$, then the new 
probability space $(\Omega, \F_{\H},\nu_{\H})$ is defined as follows: 
\begin{eqnarray*}
  &&\F_{\H}=\{\mathscr{A}\H |\mathscr{A}\subset SO_3\}\bigcap\F, \\
  &&\nu_{\H}(\mathscr{A}\H)=\nu(\mathscr{A}\H). 
\end{eqnarray*}
It can be proved that for any $\F_{\H}$ measurable function $h$, 
$$
\int_{\Omega}h\d\nu_{\H}=\int_{SO_3}\tilde{h}\d\nu
$$
where $\tilde{h}$ is defined as $\tilde{h}(P)=h(P\H)$. Hence the free energy could be rewritten as 
\begin{eqnarray*}
  \frac{F}{c}&=&\frac{F_0}{c}+k_BT\log c 
  + k_BT\left[\int_{\Omega}f(P\H)\ln f(P\H)\d\nu_{\H}\right.\\
  &&\left.+\frac{c}{2}\int_{\Omega\times\Omega}
  f(P\H)G(\H\bar{P}\H)f(P'\H)\d\nu_{\H}(P\H)\d\nu_{\H}(P'\H)\right]
\end{eqnarray*}
with the normalization condition $\int_{\Omega}f(P\H)\d\nu_{\H}=1$. 

The above process reduces the configuration space of molecules with 
$C_{\infty}$ symmetry to $S^2$. 
\begin{thm}
  For molecules with $C_{\infty}$ symmetry, the configuration space is reduced 
to $S^2$ with the uniform probablity measure 
$$
\d\nu_{\H}=\frac{\sin\alpha}{4\pi}\d\alpha\d\beta. 
$$
\end{thm}
\proof
Because $\H$ consists of all the rotations around $\bm{m}_1$, 
$$
P(\alpha,\beta,\gamma)\H=\big\{P(\alpha,\beta,\theta)|\theta\in[0,2\pi)\big\}. 
$$
So we could select $P(\alpha,\beta,0)$ as the representative element of $P\H$, 
and $\Omega$ becomes
$$
\Omega=\big\{P(\alpha,\beta,0)|\alpha\in[0,\pi],\beta\in[0,2\pi)\big\}=S^2. 
$$
Hence any $\mathscr{A}\H\in\F$ equals to $\mathscr{A}$ if $\mathscr{A}$ 
consists of some $P(\cdot,\cdot,0)$. The measure on the reduced configuration 
space is given by 
\begin{eqnarray*}
\nu_{\H}(\mathscr{A}\H)&=&\nu(\mathscr{A}\H)\\
&=&\int_{\mathscr{A}\H} \frac{\sin\alpha}{8\pi^2}\d\alpha\d\beta\d\gamma
\\
&=&\int_{\mathscr{A}}\frac{\sin\alpha}{4\pi}\d\alpha\d\beta
\int_0^{2\pi}\frac{1}{2\pi}\d\gamma\\
&=&\int_{\mathscr{A}}\frac{\sin\alpha}{4\pi}\d\alpha\d\beta. 
\end{eqnarray*}
Thus $\d\nu_{\H}=\frac{\sin\alpha}{4\pi}\d\alpha\d\beta$ is the 
uniform measure on $S^2$. 
\qed

\section{Polynomial approximation of kernel function}
\label{Trunc}
In this section we describe the construction of approximate kernel function. 
We aim to use a polynomial of nine elements of $\bar{P}$
$$
\bar{P}=(\bm{m}_i\cdot\bm{m'}_j)_{3\times 3}\triangleq(p_{ij})_{3\times 3}
$$
as the approximation. 
This form of approximation reduces the Euler-Lagrange equation (\ref{EL}) to 
a few self-consistent equations about moments. Suppose that $G$ has a term
$$
C(\bm{m}_{\sigma_1}\cdot\bm{m}_{\sigma'_1})\ldots
(\bm{m}_{\sigma_n}\cdot\bm{m}_{\sigma'_n})
=C(\bm{m}_{\sigma_1}\ldots\bm{m}_{\sigma_n})
:(\bm{m}_{\sigma'_1}\ldots\bm{m}_{\sigma'_n}). 
$$
It corresponds to a term
\begin{eqnarray*}
  &&C\int\d\nu(P)\d\nu(P')f(P)(\bm{m}_{\sigma_1}\ldots\bm{m}_{\sigma_n})
  :(\bm{m}_{\sigma'_1}\ldots\bm{m}_{\sigma'_n})f(P')\\
  &=&C\left(\int\d\nu(P)f(P)\bm{m}_{\sigma_1}\ldots\bm{m}_{\sigma_n}\right):
  \left(\int\d\nu(P')(\bm{m}_{\sigma'_1}\ldots\bm{m}_{\sigma'_n})f(P')\right)\\
  &=&C\left<\bm{m}_{\sigma_1}\ldots\bm{m}_{\sigma_n}\right>:
  \left<\bm{m}_{\sigma'_1}\ldots\bm{m}_{\sigma'_n}\right>
\end{eqnarray*}
in the free energy. And $W(P)$ must be of the form
$$
W(P)=\sum C\left<\bm{m}_{\sigma_1}\ldots\bm{m}_{\sigma_n}\right>
:\bm{m'}_{\sigma'_1}\ldots\bm{m'}_{\sigma'_n}.
$$
Denote by $\M$ the set of moments that appear in the free energy. 
This formula indicates that $W(P)$ is determined by the value of moments in 
$\M$. On the other hand, the moments can be calculated with (\ref{Boltz}) by 
\begin{equation}\label{SelfC}
  \left<\bm{m}_{\sigma_1}\ldots\bm{m}_{\sigma_n}\right>=
  C\int\d\nu(P')\bm{m'}_{\sigma_1}\ldots\bm{m'}_{\sigma_n}\exp\big(-W(P')\big). 
\end{equation}
Notice that the right side is a function of the moments. 
Applying this formula to all the moments in $\M$, we obtain a group of
self-consistent equations about these moments. So we only need to solve 
the moments in $\M$ instead of $f$. 

Next we will deduce the form of polynomial approximation of kernel function 
from its symmetries. 
Maier-Saupe potential will be derived naturally from the analysis. 
The nine elements of $\bar{P}$ are not independent. The third column is 
uniquely determined by the other two columns: 
$$
p_{i3}=p_{i+1,1}p_{i+2,2}-p_{i+1,2}p_{i+2,1},\ i=1,2,3. 
$$
Therefore $G$ can be expressed by a function of six variables
$$
p_{ij},\quad i=1,2,3,\ j=1,2. 
$$
\begin{propos}\label{SymP2}
If a molecule has reflection symmetry, and the symmetry plane is perpendicular 
to $\bm{m}_3$, 
then $G$ depends only on the following four elements of $\bar{P}$: 
$$
p_{ij},\quad i,j=1,2. 
$$
\end{propos}
\begin{proof}
  When $p_{ij}(i,j=1,2)$ are given properly, $(p_{31},p_{32})$ might take two 
  possible pairs of value: $(y_1,y_2)$ and $(-y_1,-y_2)$, which satisfy
  \begin{eqnarray*}
  &&p_{11}^2+p_{21}^2+y_1^2=p_{12}^2+p_{22}^2+y_2^2=1, \\
  &&p_{11}p_{11}+p_{21}p_{22}+y_1y_2=0. 
  \end{eqnarray*}
  Using Theorem \ref{SymP}, we have
  $$
  G(J\bar{P}J)=G(\bar{P})
  $$
  with
  \begin{equation}\label{m3}
  J=\left(
  \begin{array}{ccc}
    -1&\ 0&\ 0\\
    0&\ -1&\ 0\\
    0&\ 0&\ 1
  \end{array}
  \right).
  \end{equation}
  Note that $J\bar{P}J$ leaves $p_{ij}(i,j=1,2)$ remained, 
  but changes the sign of $p_{31},p_{32}$, thus
  $$
  G(p_{11},p_{21},y_1,p_{12},p_{22},y_2)=G(p_{11},p_{21},-y_1,p_{12},p_{22},-y_2)
  =G(p_{11},p_{21},p_{12},p_{22}). 
  $$
\end{proof}
Next we examine the properties of $G$ when $\H$ contains a rotation of an angle 
$\theta$ around $\bm{m}_1$. The matrix that represents this rotation is
\begin{equation}
J_{\theta}=\left(
\begin{array}{ccc}
  1&\ 0&\ 0\\
  0&\ \cos\theta&\ -\sin\theta\\
  0&\ \sin\theta&\ \cos\theta
\end{array}
\right).
\end{equation}
Direct computation gives
\begin{equation}\label{Jtheta}
J_{\theta}\bar{P}(\alpha,\beta,\gamma)=\bar{P}(\alpha,\beta+\theta,\gamma),\quad
\bar{P}(\alpha,\beta,\gamma)J_{\theta}=\bar{P}(\alpha,\beta,\gamma+\theta)
\end{equation}
where $\bar{P}(\alpha,\beta,\gamma)$ is the representation of $\bar{P}$
by Euler angles.

\begin{propos}\label{Rot}
If $J_{\theta}\in\H$, then 
\begin{equation}\label{Rot0}
  G\big(\bar{P}(\alpha,\beta,\gamma)\big)=G\big(\bar{P}(\alpha,\beta+\theta,\gamma)\big)
  =G\big(\bar{P}(\alpha,\beta,\gamma+\theta)\big). 
\end{equation}
\end{propos}
\proof
Using Theorem \ref{TInv}, we have
\begin{eqnarray*}
  G(\bar{P})=G(\bar{P}J_{\theta}),\qquad
  G(\bar{P})=G(J_{\theta}\bar{P}). 
\end{eqnarray*}
Along with (\ref{Jtheta}), we obtain (\ref{Rot0}). 
\qed

In the following theorem, $\bm{m}_1$ always coincides with the rotational axis. 
\begin{thm}\label{SymB}
\begin{enumerate}
\item 
For a molecule with $C_{2v}$ symmetry, $G$ is a function of $p_{11},p_{12},p_{21},p_{22}$, with
\begin{equation}\label{Sym_m1}
  G(p_{11},p_{12},p_{21},p_{22})=G(p_{11},-p_{12},p_{21},-p_{22})
  =G(p_{11},p_{12},-p_{21},-p_{22}). 
\end{equation}
\item
For a molecule with $C_{\infty}$ symmetry, $G$ is a function of 
$p_{11}=\bm{m}_1\cdot\bm{m'}_1$. If the molecule has $D_{\infty h}$ symmetry, 
$G$ is a function of $|p_{11}|$. 
\end{enumerate}
\end{thm}
\proof
\begin{enumerate}
\item 
Theorem \ref{TInv} gives $G(\bar{P})=G(J_{\pi}\bar{P})=G(\bar{P}J_{\pi})$. 
By Proposition \ref{SymP2}, (\ref{Sym_m1}) holds. 
\item
  Axially symmetry means that Proposition \ref{Rot} is valid with arbitrary 
  $\theta$. Therefore
  \begin{equation}\label{AxiSym1}
  G\big(\bar{P}(\alpha,\beta,\gamma)\big)=G\big(\bar{P}(\alpha,0,0)\big)
  =G(\cos\alpha)=G(p_{11}). 
  \end{equation}
  Like in Theorem \ref{SymP}, we suppose that the plane contains $\hat{O}$. 
  Note that the rotation of $\pi$ around $\bm{m}_3$, represented by $J$ defined 
  in (\ref{m3}), is contained in $D_{\infty h}$. By Theorem \ref{TInv}, 
  $G(\bar{P}J)=G(\bar{P})$. We deduce from (\ref{AxiSym1}) that
  $$
  G(\bar{P})=G(p_{11})=G(-p_{11})=G(|p_{11}|). 
  $$
\end{enumerate}
\qed

With the above discussion, we are able to construct polynomial approximations 
for molecules with different symmetries. We start from the approximate 
kernel function of molecules with $D_{\infty h}$ symmetry. 
In Theorem \ref{SymB} we have proven that 
$G=G(|\bm{m}_1\cdot\bm{m'}_1|)$. Its approximation should be a polynomial of 
$\bm{m}_1\cdot\bm{m'}_1$ without odd-degree terms. Therefore it is at least
quadratic, which coincides with the form of Maier-Saupe potential: 
$$
G=c_0+c_2(\bm{m}_1\cdot\bm{m'}_1)^2. 
$$
The above form indicates that $\M=\big\{\left<\bm{m}_1\bm{m}_1\right>\big\}$. 

When a molecule has only $C_{\infty}$ symmetry, odd-degree terms of $p_{11}$ 
no longer vanishes. Quadratic approximation will be
\begin{equation}\label{PolRod}
G=c_0+c_1(\bm{m}_1\cdot\bm{m'}_1)+c_2(\bm{m}_1\cdot\bm{m'}_1)^2, 
\end{equation}
which is discussed in \cite{Dipol}. When this kernel is used, 
$\M=\big\{\left<\bm{m}_1\right>,\ \left<\bm{m}_1\bm{m}_1\right>\big\}$. 

Now we turn to the approximations of kernel function for molecules with 
$C_{2v}$ symmetry, including bent-core molecules and spherotriangles.
By Proposition \ref{SymP2}, an approximation of $G$ is a polynomial of four 
variables $p_{11},p_{12},p_{21},p_{22}$. Then by Proposition \ref{RelP}, it is 
symmetric with respect to $p_{12}$ and $p_{21}$, which means
$$
G=G(p_{11},p_{22},p_{12}+p_{21},p_{12}p_{21}). 
$$
Using (\ref{Sym_m1}), we are able to determine the form of polynomial. 
Quadratic approximation is written as
\begin{equation}\label{QuaApp}
  G=c_0+c_1p_{11}+c_2p_{11}^2+c_3p_{22}^2+c_4(p_{12}^2+p_{21}^2). 
\end{equation}
Cubic approximation is written as 
\begin{equation}\label{CubApp}
  G=c_0+c_1p_{11}+c_2p_{11}^2+c_3p_{22}^2+c_4(p_{12}^2+p_{21}^2)
  +c_5p_{11}^3+c_6p_{11}p_{22}^2+c_7p_{11}(p_{12}^2+p_{21}^2)+c_8p_{12}p_{21}p_{22}.
\end{equation}
For quadratic approximation, 
$$
\M=\big\{\left<\bm{m}_1\right>, 
\left<\bm{m}_1\bm{m}_1\right>, \left<\bm{m}_2\bm{m}_2\right>\big\}; 
$$
for cubic approximation, 
$$
\M=\big\{\left<\bm{m}_1\right>, 
\left<\bm{m}_1\bm{m}_1\right>, \left<\bm{m}_2\bm{m}_2\right>, 
\left<\bm{m}_1\bm{m}_1\bm{m}_1\right>,\left<\bm{m}_1\bm{m}_2\bm{m}_2\right>\big\}.
$$

From the above discussion we know that the form of polynomial approximation 
is determined by molecular symmetries. 
The coefficients $c_i$ can be calculated by projecting $G$ to the space 
spanned by all the polynomials of the given form. If the approximation has 
the form
$$
\sum_i c_iq_i(\bar{P}), 
$$
then the coefficients $c_i$ are determined by 
$$
\sum_i \left[\int_{SO_3}\d \nu(\bar{P})q_i(\bar{P})q_j(\bar{P})\right]c_i
=\int_{SO_3}\d \nu(\bar{P})G(\bar{P};\Theta)q_j(\bar{P}). 
$$
In the above, $\Theta$ is a set that consists of temperature and a group of 
molecular parameters. The formula reveals that these coefficients are functions 
of $\Theta$. Generally speaking, as temperature is included in $\Theta$, the 
approximate $G$ is able to describe both lyotropic and thermotropic liquid 
crystals. 

The projection of Onsager potential to span$\{1,p_{11}^2\}$ gives
\begin{equation}
  c_2=-\frac{15\pi}{32}cL^2D. 
\end{equation}
As a constant difference in $G$ does not affect the solution, $c_0$ is ignored. 
It is easy to see that $c_2$ is propotional to one effective parameter $cL^2D$. 
The projection of the excluded-volume potential of spherocuboids is derived by 
R. Rosso and E. G. Virga in \cite{quadproj}. In Appendix, we discuss the 
projection of excluded-volume potential of isoceles spherotriangles to the 
space of quadratic approximations. Suppose that the top corner is $\theta$ and 
the length of lateral is $L/2$. The results are
\begin{eqnarray}
  c_2&=&-\frac{15}{64}cL^3\sin\theta\cos^2\frac{\theta}{2}
  -\frac{15\pi}{128}cL^2D\cos^4\frac{\theta}{2},\label{c_2}\\
  c_3&=&-\frac{15}{64}cL^3\sin\theta\sin\frac{\theta}{2}(1+\sin\frac{\theta}{2})
  -\frac{15\pi}{128}cL^2D\sin^2\frac{\theta}{2}(1+\sin\frac{\theta}{2})^2, \label{c_3}\\
  c_4&=&-\frac{15}{128}cL^3\sin\theta(1+\sin\frac{\theta}{2})
  -\frac{15\pi}{128}cL^2D\cos^2\frac{\theta}{2}\sin\frac{\theta}{2}
  (1+\sin\frac{\theta}{2}). \label{c_4}
\end{eqnarray}
And $c_1$ is proportional to $cL^2D$ with
$$
c_1=\frac{3}{8}cL^2DK(\theta), 
$$
where $K(\theta)$ is a function of $\theta$ defined in (\ref{Ktheta}). 

\section{Further analysis and the choice of order parameters}\label{Ord}
In the previous section we select some moments and reduce the density 
functional theory to a group of equations about them. Those equations usually 
imply some properties of the moments. They could help us to choose independent 
components of the moments as order parameters. Here we try to extract these 
properties. Some of them depend on the values of coefficients. As the 
coefficients are determined by molecular parameters and temperature, 
these properties would reveal the impacts of them. 

When $G$ takes Maier-Saupe potential, the only moment in $\M$ is 
$\left<\bm{m}_1\bm{m}_1\right>$. It can be diagonalized by selecting 
axes along its eigenvectors. Its trace equals to $1$, leaving only two degrees 
of freedom remained. These two degrees of freedom could be further reduced to 
one by the proof of uniaxial property\cite{AxiSymMS,AxiSym2,AxiSym3}. 
\begin{thm}[Axially symmetry of the solution with Maier-Saupe potential]
If $G$ takes Maier-Saupe potential (\ref{MS}), every solution of (\ref{EL}) is 
axially symmetric
\begin{equation}\label{Uniaxi}
  f=f(\bm{m}_1\cdot\bm{n})=C\exp(-\eta(\bm{m}_1\cdot\bm{n})^2). 
\end{equation}
\end{thm}

When $\M$ has more than one moments, there are usually some relations between 
them. In \cite{Dipol} the following conclusion is shown, 
which reduces the number of order parameters for polar rods to $3$.
\begin{thm}\label{m1_0_org}
When $G$ takes (\ref{PolRod}), we have 
\begin{enumerate}
\item $\left<\bm{m}_1\right>$ parallels to one of the eigenvectors of 
$\left<\bm{m}_1\bm{m}_1\right>$. 
\item If $-c_1\le 1$ in (\ref{PolRod}), $\left<\bm{m}_1\right>=0$. 
\end{enumerate}
\end{thm}

From now on we will focus on the kernel (\ref{QuaApp}):
$$
  G=c_1p_{11}+c_2p_{11}^2+c_3p_{22}^2+c_4(p_{12}^2+p_{21}^2)
$$
where $c_0$ is set to zero, for it does not affect the solutions. 
$W(P)$ is written as
\begin{eqnarray*}
W(P)&=&c_1\left<\bm{m}_1\right>\cdot\bm{m}_1
+\big(2c_2\left<\bm{m}_1\bm{m}_1\right>+c_4\left<\bm{m}_2\bm{m}_2\right>\big)
:\bm{m}_1\bm{m}_1\\
&&+\big(2c_3\left<\bm{m}_2\bm{m}_2\right>+c_4\left<\bm{m}_1\bm{m}_1\right>\big)
:\bm{m}_2\bm{m}_2\\
&=&c_1\left<\bm{m}_1\right>\cdot\bm{m}_1+W_1(P). 
\end{eqnarray*}
Write down the components of $\bm{m}_i$ as
$$
(\bm{m}_1,\bm{m}_2,\bm{m}_3)=(\bm{e}_1,\bm{e}_2,\bm{e}_3)\left(
\begin{array}{ccc}
m_{11}&m_{21}&m_{31}\\
m_{12}&m_{22}&m_{32}\\
m_{13}&m_{23}&m_{33}
\end{array}
\right).
$$
Recalling the equality (\ref{RotP}), we know that $m_{ij}$ are the elements of 
$P$. The next theorem contains a direct extension of the second part of Theorem 
\ref{m1_0_org}, and discusses the relationship of axes of three moments 
$\left<\bm{m}_{1}\right>,\left<\bm{m}_{1}\bm{m}_{1}\right>,
\left<\bm{m}_{2}\bm{m}_{2}\right>$. 

\begin{thm}\label{MoB}
\begin{enumerate}
\item If $-c_1\le 1$, $\left<\bm{m}_1\right>=0$. 
\item If $\left<\bm{m}_1\bm{m}_1\right>$ and $\left<\bm{m}_2\bm{m}_2\right>$ 
can be diagonalized simutaneously, $\left<\bm{m}_1\right>$ parallels to one of 
the eigenvectors of $\left<\bm{m}_1\bm{m}_1\right>$. 
\item If $c_1\ge -1$ and $c_4^2=c_2c_3$, $\left<\bm{m}_1\bm{m}_1\right>$ and 
$\left<\bm{m}_2\bm{m}_2\right>$ can be diagonalized simutaneously by the axes 
of $\left<d_1\bm{m}_1\bm{m}_1+d_2\bm{m}_2\bm{m}_2\right>$, where 
$$
c_2=\pm d_1^2,\ c_3=\pm d_2^2,\ c_4=d_1d_2. 
$$
\item If $\left<d_1\bm{m}_1\bm{m}_1+d_2\bm{m}_2\bm{m}_2\right>$ is uniaxial, and 
$\left<\bm{m}_1\right>$ parallels to the axis, then both 
$\left<\bm{m}_1\bm{m}_1\right>$ and $\left<\bm{m}_2\bm{m}_2\right>$ are 
uniaxial. 
\end{enumerate}
\end{thm}
\proof
\begin{enumerate}
\item 
Set $J=\mbox{diag}(-1,1,-1)$. It is easy to verify that
$$
W_1(PJ)=W_1(P). 
$$
The self-consistent equation of $\bm{m}_1$ yields
\begin{eqnarray*}
\left<\bm{m}_{1}\right>
&=&\frac{2\int \d\nu \bm{m}_{1}
\exp\big(-W_1(P)-c_1\left<\bm{m}_{1}\right>\cdot\bm{m}_{1}\big)}
{2\int\d\nu\exp\big(-W_1(P)-c_1\left<\bm{m}_{1}\right>\cdot\bm{m}_{1}\big)}\\
&=&\frac{\int \d\nu \bm{m}_{1}\big[
\exp\big(-W_1(P)-c_1\left<\bm{m}_{1}\right>\cdot\bm{m}_{1}\big)
-\exp\big(-W_1(JP)+c_1\left<\bm{m}_{1}\right>\cdot\bm{m}_{1}\big)\big]}
{\int\d\nu\big{[}\exp(-W_1(P)-c_1\left<\bm{m}_{1}\right>\cdot\bm{m}_{1})
+\exp(-W_1(JP)+c_1\left<\bm{m}_{1}\right>\cdot\bm{m}_{1})\big{]}}\\
&=&\frac{\int \d\nu \bm{m}_{1}\exp\big(-W_1(P)\big)
\sinh\big(-c_1\left<\bm{m}_{1}\right>\cdot\bm{m}_{1}\big)}
{\int \d\nu \exp(-W_1(P))
\cosh\big(-c_1\left<\bm{m}_{1}\right>\cdot\bm{m}_{1}\big)}. 
\end{eqnarray*}
Therefore
\begin{equation}\label{m1m}
|\left<\bm{m}_{1}\right>|^2
=\frac{\int \d\nu \exp(-W_1(P))\left<\bm{m}_{1}\right>\cdot\bm{m}_{1}
\sinh\big(-c_1\left<\bm{m}_{1}\right>\cdot\bm{m}_{1}\big)}
{\int \d\nu \exp(-W_1(P))
\cosh\big(-c_1\left<\bm{m}_{1}\right>\cdot\bm{m}_{1}\big)}. 
\end{equation}
If $-c_1\le 0$, then 
$$
\left<\bm{m}_{1}\right>\cdot\bm{m}_{1}
\sinh(-c_1\left<\bm{m}_{1}\right>\cdot\bm{m}_{1})\le 0, 
$$
which yields $|\left<\bm{m}_{1}\right>|^2\le 0$. 
If $0<-c_1\le 1$, using $x\tanh(x)<x^2$ for $x\ne 0$, we get
\begin{eqnarray*}
&&\left<\bm{m}_{1}\right>\cdot\bm{m}_{1}
\sinh\big(-c_1\left<\bm{m}_{1}\right>\cdot\bm{m}_{1}\big)\\
&=& -c_1\left<\bm{m}_{1}\right>\cdot\bm{m}_{1}
\tanh\big(-c_1\left<\bm{m}_{1}\right>\cdot\bm{m}_{1}\big)
\cosh\big(-c_1\left<\bm{m}_{1}\right>\cdot\bm{m}_{1}\big)\\
&<&-c_1\big(\left<\bm{m}_{1}\right>\cdot\bm{m}_{1}\big)^2
\cosh\big(-c_1\left<\bm{m}_{1}\right>\cdot\bm{m}_{1}\big)\\
&\le& |\left<\bm{m}_{1}\right>|^2|\bm{m}_{1}|^2
\cosh\big(-c_1\left<\bm{m}_{1}\right>\cdot\bm{m}_{1}\big)\\
&=&|\left<\bm{m}_{1}\right>|^2
\cosh\big(-c_1\left<\bm{m}_{1}\right>\cdot\bm{m}_{1}\big).
\end{eqnarray*}
If $\left<\bm{m}_1\right>\ne 0$, 
we substitute the above inequality into (\ref{m1m}) and get 
$|\left<\bm{m}_{1}\right>|^2<|\left<\bm{m}_{1}\right>|^2$, which is a contradiction. 
\item
Select coordinate axes that diagonalize $\left<\bm{m}_1\bm{m}_1\right>$ and 
$\left<\bm{m}_2\bm{m}_2\right>$. Now $W_1(P)$ is of the form
$$
W_1(P)=\sum_{i=1,2,j=1,2,3} c_{ij}m_{ij}^2.
$$
Set 
\begin{equation}\label{Js}
J_1=\mbox{diag}(-1,1,1),\  J_2=\mbox{diag}(1,-1,1),\ J_3=\mbox{diag}(1,1,-1),
\end{equation}
the form of $W_1(P)$ indicates that
\begin{equation}\label{Ws}
W_1(P)=W_1(J_1PJ_3)=W_1(J_2PJ_3)=W_1(J_1J_2P).
\end{equation}
If $\left<\bm{m}_1\right>$ does not parallel to any one of the axes, at least 
two of its components are nonzero. 
Suppose $r_1=\left<m_{11}\right>\ne 0$, $r_2=\left<m_{12}\right>\ne 0$ and 
denote $r_3=\left<m_{13}\right>$. 
By part 1 of the current theorem, $-c_1>1$. 
Thus $x\sinh(-c_1x)>0$ for $x\ne 0$. 
Using the self-consistent equation of $\left<m_{11}m_{12}\right>$, we get
\begin{eqnarray*}
&&4r_1r_2\left<m_{11}m_{12}\right>\\
&=&\frac{4}{Z}\int \d\nu 
r_1r_2m_{11}m_{12}\exp\big{(}-W_1(P)-c_1(r_1m_{11}+r_2m_{12}+r_3m_{13})\big{)}\\
&=&\frac{1}{Z}\int \d\nu r_1r_2m_{11}m_{12}\exp(-c_1r_3m_{13})
\Big{[}\exp\big(-W_1(P)-c_1(r_1m_{11}+r_2m_{12})\big)\\
&&-\exp\big(-W_1(J_1PJ_3)-c_1(-r_1m_{11}+r_2m_{12})\big)\\
&&-\exp\big(-W_1(J_2PJ_3)-c_1(r_1m_{11}-r_2m_{12})\big)\\
&&+\exp\big(-W_1(J_1J_2P)-c_1(-r_1m_{11}-r_2m_{12})\big)\Big{]}\\
&=&\frac{1}{Z}\int \d\nu r_1r_2m_{11}m_{12}\exp(-W_1(P)-c_1r_3m_{13})
\sinh (-c_1r_1m_{11})\sinh(-c_2r_2m_{12})\\
&>&0.
\end{eqnarray*}
This inequality violates the diagonalization of $\left<\bm{m}_1\bm{m}_1\right>$.
\item
From $c_4^2=c_2c_3$, we can write
$$
c_2p_{11}^2+c_3p_{22}^2+c_4(p_{12}^2+p_{21}^2)=\pm 
(d_1\bm{m}_1\bm{m}_1+d_2\bm{m}_2\bm{m}_2):(d_1\bm{m'}_1\bm{m'}_1+d_2\bm{m'}_2\bm{m'}_2).
$$
Without loss of generality, the sign on the right side is assumed positive. 
Because $c_1\ge -1$, we get $\left<\bm{m}_1\right>=0$. Thereby $W$ converts into
\begin{eqnarray*}
W(P)=W_1(P)=\left<d_1\bm{m}_1\bm{m}_1+d_2\bm{m}_2\bm{m}_2\right>:(d_1\bm{m}_1\bm{m}_1+d_2\bm{m}_2\bm{m}_2). 
\end{eqnarray*}
This would enable us to select coordinate axes according to the axes of 
$\left<d_1\bm{m}_1\bm{m}_1\right.+\left.d_2\bm{m}_2\bm{m}_2\right>$. 
We show that 
$\left<\bm{m}_1\bm{m}_1\right>$ and $\left<\bm{m}_2\bm{m}_2\right>$ are 
diagonalized as well. In other words, we need to show that the off-diagonal 
elements are zero. Let $J_1,J_2,J_3$ be defined as in (\ref{Js}). 
Equation (\ref{Ws}) still holds for $W_1$. Therefore 
\begin{eqnarray*}
&&4\left<m_{11}m_{12}\right>\\
&=&\frac{4}{Z}\int \d\nu m_{11}m_{12}\exp(-W_1(P))\\
&=&\frac{1}{Z}\int \d\nu m_{11}m_{12}
\Big[\exp\big(-W_1(P)\big)-\exp\big(-W_1(J_1PJ_3)\big)\\
&&-\exp\big(-W_1(J_2PJ_3)\big)+\exp\big(-W_1(J_1J_2P)\big)\Big]\\
&=&0. 
\end{eqnarray*}
The zero values of other off-diangonal elements can be obtained similarly. 
\item
First we choose axes that diagonalize $\left<d_1\bm{m}_1\bm{m}_1
+d_2\bm{m}_2\bm{m}_2\right>$ with diagonal elements $b_1,b_2,b_3$. 
The uniaxiality requires that two of them are equal. Assume $b_1=b_2$, 
then $\left<m_{11}\right>=\left<m_{12}\right>=0$. 
Thereby $W$ is simplified to
\begin{eqnarray*}
W(P)&=&d_1(b_1(m_{11}^2+m_{12}^2)+b_3m_{13}^2)+d_2(b_1(m_{21}^2+m_{22}^2)+b_3m_{23}^2)+r_{13}m_{13}\\
&=&(d_1+d_2)b_1+d_1(b_3-b_1)m_{13}^2+d_2(b_3-b_1)m_{23}^2+r_{13}m_{13}
\end{eqnarray*}
where $r_{13}=\left<m_{13}\right>$. Set 
$$
J=\left(
\begin{array}{ccc}
  0&\ -1&\ 0\\
  1&\ 0&\ 0\\
  0&\ 0&\ 1
\end{array}
\right). 
$$
Using $W(P)=W(JP)$, we get
\begin{eqnarray*}
\left<m_{12}^2\right>&=&\frac{1}{Z}\int \d\nu m_{12}^2\exp\big(-W(P)\big)\\
&=&\frac{1}{Z}\int \d\nu m_{11}^2\exp\big(-W(JP)\big)\\
&=&\frac{1}{Z}\int \d\nu m_{11}^2\exp\big(-W(P)\big)\\
&=&\left<m_{11}^2\right>. 
\end{eqnarray*}
\end{enumerate}
\qed

We tend to believe that $\left<\bm{m}_1\bm{m}_1\right>$ and 
$\left<\bm{m}_2\bm{m}_2\right>$ can be diagonalized simutaneously. 
The results of Theorem \ref{MoB} would reduce the degrees of freedom 
of order parameters of bent-core molecules to $5$. We choose coordinate axes as 
eigenvectors of $\left<\bm{m}_1\bm{m}_1\right>$ as well as those of 
$\left<\bm{m}_2\bm{m}_2\right>$, then $\left<\bm{m}_1\right>$ parallels to one 
of the eigenvectors of $\left<\bm{m}_1\bm{m}_1\right>$. Because the trace of 
$\left<\bm{m}_1\bm{m}_1\right>$ and $\left<\bm{m}_2\bm{m}_2\right>$ equal to 
$1$, both of the two second moments contribute two degrees of freedom. At last 
$\left<\bm{m}_1\right>$ contributes one, making them five in total. 

Finally, we should point out that the set of order parameters should be decided 
by results of experiments and simulations so as to be able to distinguish 
dirrerent phases. It should also follow this criterion to determine where to 
truncate the polynomial approximation of $G$. 
Rod-like molecules exhibits only uniaxial nematics. 
As we have described, Maier-Saupe potential is a polynomial approximation of 
$G$ truncated on the second order. With thorough analysis the number of order 
parameter is reduced to 1. Therefore Maier-Saupe potential is proven to be the 
most concise model of rod-like molecules that covers experimental results. 
Up to now, spatially homogeneous phases of bent-core molecule are restrained to 
uniaxial or biaxial nematics, without the observation of polar order. 
This seems to indicate the sufficiency to approximate $G$ with quadratic 
polynomials, which contradicts with what is proposed in \cite{OrdPar}. Also it 
will be interesting to see if any phases with polar order would appear. 

\section{Conclusion and outlook}\label{Con}
A generic modelling procedure is proposed for rigid molecules of arbitrary 
shape. The modelling of kernel function incorporates pairwise interaction. 
We show that the symmetries of molecule determine the reduced configuration 
space and the form of polynomial approximations of $G$ with its coefficients 
depended on temperature and molecular parameters. 
An approximate kernel is deduced for molecules with $C_{2v}$ symmetry. 
By approximating $G$ with polynomial, the system is reduced to a 
group of equations about moments of body-fixed axes. Some properties of these 
moments are studied for molecules with $C_{2v}$ symmetry, 
and the number of order parameters is predicted for bent-core molecules. 
The prediction needs to be verified by results of simulations and comparison to
experiments. Moreover, it remains unknown whether there are some general 
relationships between the moments. 
A clear understanding of them would help us to find out a
minimal complete set of order parameters. 

\section{Appendix}
\subsection{The excluded-volume potential of spherotriangles}
\subsubsection{The calculation of excluded volume}
\begin{figure}
\centering
\includegraphics[width=0.8\textwidth,keepaspectratio]{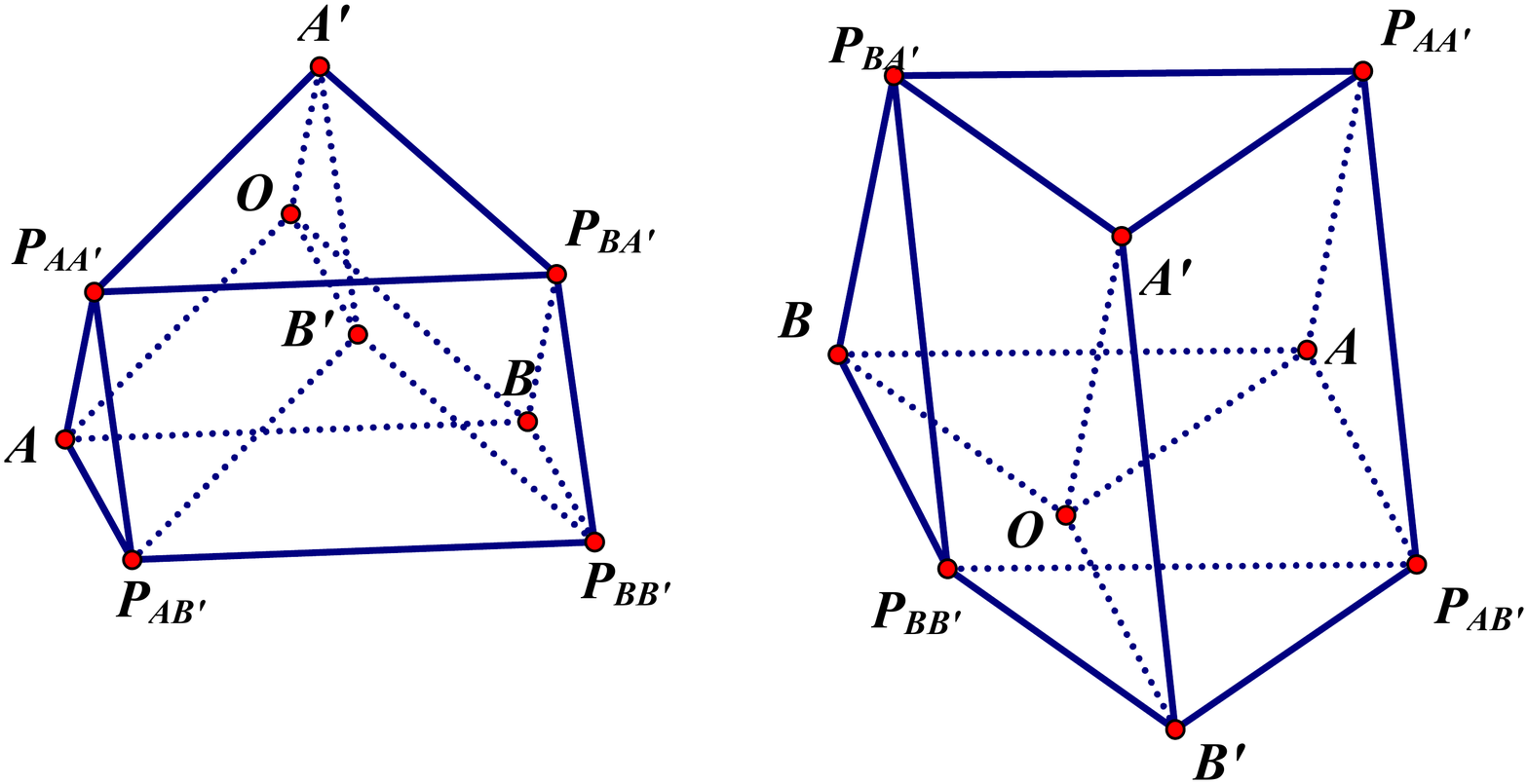}
\caption{Two possible geometries of $K$}\label{tri}
\end{figure}
Here we calculate the excluded-volume potential of two 
spherotriangles $T_1+B_{D/2}$ and $T_2+B_{D/2}$. 
The excluded region can be represented by $K+B_D$ where $K=T_1-T_2$. 
By (\ref{Steiner}), we need to calculate $V_3,\ V_2,\ V_1$ for $K$. 
Denote the vertices of $T_1$ as $OAB$ that lie in plane $\pi$, 
and those of $-T_2$ as $O'A'B'$ that lie in plane $\pi'$. 
$K$ is a polytope, for it is the convex hull of nine points
$$
\{O,A,B\}+\{O',A',B'\}. 
$$
The edges of two triangles are denoted as
$$
\overrightarrow{AO}=\bm{a},\ \overrightarrow{OB}=\bm{b},\ \overrightarrow{BA}=\bm{c},\ 
\overrightarrow{A'O'}=\bm{a'},\ \overrightarrow{O'B'}=\bm{b'},\ \overrightarrow{B'A'}=\bm{c'}.
$$

If $\pi$ and $\pi'$ do not parallel, we can label the vertices properly such 
that the plane $\pi+O'-O$ seperates $A'$ and $B'$, and the plane $\pi'+O-O'$ 
seperates $A$ and $B$, namely 
\begin{eqnarray*}
  (\bm{m}_3\cdot\bm{a'})(\bm{m}_3\cdot\bm{b'})\ge 0,\\
  (\bm{m'}_3\cdot\bm{a})(\bm{m'}_3\cdot\bm{b})\ge 0.
\end{eqnarray*}
If the intersection of two triangles $T_1$ and $T_2+O-O'$ is not empty, which 
indicates
$$
(\bm{m}_3\cdot\bm{c'})(\bm{m'}_3\cdot\bm{c})<0,
$$
$K$ is drawn in the left part of Fig.\ref{tri}; 
and if it is empty, which indicates
$$
(\bm{m}_3\cdot\bm{c'})(\bm{m'}_3\cdot\bm{c})>0,
$$
Now we may assume that $O=O'$, then $K$ is drawn in the right part of 
Fig.\ref{tri}. The notion $P_{AA'}$ represents the point located at 
$O+\overrightarrow{OA}+\overrightarrow{O'A'}$, etc..
When $\pi$ and $\pi'$ are parallel, we can label the vertices such that 
$T_2$ intersects with $\angle AOB$ or its vertical angle. 

First we calculate $V_3(K)$. For the case on the left part of Fig.\ref{tri}, 
$K$ can be divided into the prisms
$$
AP_{AA'}P_{AB'}-BP_{BA'}P_{BB'},\ A'P_{AA'}P_{BA'}-OAB,\ OAB-B'P_{AB'}P_{BB'}, 
$$
or
$$
A'P_{AA'}P_{BA'}-B'P_{AB'}P_{BB'},\ AP_{AA'}P_{AB'}-OA'B',\ OA'B'-BP_{BA'}P_{BB'}.
$$
Thus
$$
V_3=|\bm{a}\times\bm{b}\cdot\bm{a'}|+|\bm{a}\times\bm{b}\cdot\bm{b'}|
+|\bm{a'}\times\bm{b'}\cdot\bm{c}|
=|\bm{a}\times\bm{b}\cdot\bm{c'}|+|\bm{a'}\times\bm{b'}\cdot\bm{a}|
+|\bm{a'}\times\bm{b'}\cdot\bm{b}|.
$$
For the case on the right, $K$ can be divided into the prisms
$$
AP_{AA'}P_{AB'}-BP_{BA'}P_{BB'},\ A'P_{AA'}P_{BA'}-B'P_{AB'}P_{BB'},
$$
or 
$$
AP_{AA'}P_{AB'}-OA'B',\ A'P_{AA'}P_{BA'}-OAB,\ OAB-B'P_{AB'}P_{BB'},\ OA'B'-BP_{BA'}P_{BB'}. 
$$
Thus 
$$
V_3=|\bm{a}\times\bm{b}\cdot\bm{a'}|+|\bm{a}\times\bm{b}\cdot\bm{b'}|
+|\bm{a'}\times\bm{b'}\cdot\bm{a}|+|\bm{a'}\times\bm{b'}\cdot\bm{b}|
=|\bm{a}\times\bm{b}\cdot\bm{c'}|+|\bm{a'}\times\bm{b'}\cdot\bm{c}|.
$$
For both cases, we have
\begin{equation}
V_3(K)=
\frac{1}{2}\Big(|\bm{a}\times\bm{b}\cdot\bm{a'}|+|\bm{a}\times\bm{b}\cdot\bm{b'}|
+|\bm{a}\times\bm{b}\cdot\bm{c'}|+|\bm{a'}\times\bm{b'}\cdot\bm{a}|
+|\bm{a'}\times\bm{b'}\cdot\bm{b}|+|\bm{a'}\times\bm{b'}\cdot\bm{c}|\Big).
\end{equation}

Next we calculate $V_1(K)$. Each edge of $K$ parallels to one of the six edges 
of $T_1$ and $T_2$. As an example, we describe the contribution to $V_1$ of 
edges parallel to $\bm{a}$. As the faces contain one of those edges, the 
outward normal vectors lie in a plane perpendicular to $\bm{a}$. 
For the case on the left, there are three edges parallel to $\bm{a}$: 
$$
A'P_{AA'},\ OA,\ B'P_{AB'}. 
$$
As their length equals to $|\bm{a}|$, we only need to calculate the sum of 
external angles, which is
$$
\frac{1}{2\pi}(\angle \left<\bm{n}_{A'P_{AA'}P_{BA'}},\bm{n}_{OAP_{AA'}A'}\right>
+\angle \left<\bm{n}_{OAP_{AA'}A'},\bm{n}_{OAP_{AB'}B'}\right>
+\angle \left<\bm{n}_{OAP_{AB'}B'},\bm{n}_{B'P_{AB'}P_{BB'}}\right>). 
$$
Note that $\bm{n}_{A'P_{AA'}P_{BA'}}$ and $\bm{n}_{B'P_{AB'}P_{BB'}}$ are reverse, and 
the four vectors 
$$
\bm{n}_{A'P_{AA'}P_{BA'}},\ \bm{n}_{OAP_{AA'}A'},\ 
\bm{n}_{OAP_{AB'}B'},\ \bm{n}_{B'P_{AB'}P_{BB'}}
$$
are sequentially arranged. Thus the three angles add up to $\pi$, 
and the sum of the external angles equals to $\frac{1}{2}$. 
For the case on the right, there are two edges parallel to $\bm{a}$: 
$$
A'P_{AA'},\ B'P_{AB'}.
$$
Again we only need the sum of the external angles: 
$$
\frac{1}{2\pi}
(\angle \left<\bm{n}_{A'P_{AA'}P_{BA'}},\bm{n}_{A'P_{AA'}P_{AB'}B'}\right>
+\angle \left<\bm{n}_{A'P_{AA'}P_{AB'}B'},\bm{n}_{B'P_{AB'}P_{BB'}}\right>)
=\frac{1}{2}.
$$
Therefore the amount of the external angles at the edges parallel to $\bm{a}$ is always 
$\frac{1}{2}$. The above calculation can be done for the other five edges, leading to 
\begin{equation}
V_1(K)=\frac{1}{2}\Big(|\bm{a}|+|\bm{b}|+|\bm{c}|+|\bm{a'}|+|\bm{b'}|+|\bm{c'}|\Big).
\end{equation}

The expression of $V_2(K)$ is different for two cases in Fig.\ref{tri}. 
The faces of $K$ always contain four triangles $\triangle AP_{AA'}P_{AB'},
\triangle BP_{BA'}P_{BB'}, \triangle A'P_{AA'}P_{BA'}$ and $\triangle B'P_{AB'}P_{BB'}$.
The other faces are some parallelograms. For the case in the left, they are
$$
OAP_{AA'}A',\ OBP_{BA'}A',\ OAP_{AB'}B',\ OBP_{BB'}B',\ P_{AA'}P_{AB'}P_{BB'}P_{BA'}.
$$
Thus
$$
V_2(K)=|\bm{a}\times \bm{b}|+|\bm{a'}\times \bm{b'}|
+|\bm{a}\times \bm{a'}|+|\bm{a}\times \bm{b'}|+|\bm{b}\times \bm{a'}|
+|\bm{b}\times \bm{b'}|+|\bm{c}\times \bm{c'}|.
$$
For the case in the right, they are
$$
ABP_{BA'}P_{AA'},\ ABP_{BB'}P_{AB'},\ A'B'P_{AB'}P_{AA'},\ A'B'P_{BB'}P_{BA'}.
$$
Thus
$$
V_2(K)=|\bm{a}\times \bm{b}|+|\bm{a'}\times \bm{b'}|
+|\bm{c}\times \bm{a'}|+|\bm{c}\times \bm{b'}|
+|\bm{a}\times \bm{c'}|+|\bm{b}\times \bm{c'}|.
$$
We point out that 
\begin{equation}\label{PntRef}
V_2(T_1-T_2)+V_2(T_1+T_2)=
\sum_{\bm{e}\in\{\bm{a,b,c}\},\bm{e'}\in\{\bm{a',b',c'}\}}|\bm{e}\times\bm{e'}|
+2\big(|\bm{a}\times\bm{b}|+|\bm{a'}\times\bm{b'}|\big). 
\end{equation}
In fact, when $T_2$ is substituted with $-T_2$, 
$\bm{a'},\bm{b'},\bm{c'}$ convert into $-\bm{a'},-\bm{b'},-\bm{c'}$. 
So $(\bm{m}_3\cdot\bm{a'})(\bm{m}_3\cdot\bm{b'})$ and 
$(\bm{m'}_3\cdot\bm{a})(\bm{m'}_3\cdot\bm{b})$ remain unchanged, while 
$(\bm{m}_3\cdot\bm{c'})(\bm{m'}_3\cdot\bm{c})$ alters its sign. 
This means that one of $T_1-T_2$ and $T_1+T_2$ corresponds to the case in the 
left, while the other corresponds to the case in the right. Therefore 
(\ref{PntRef}) holds. 

The excluded volume of rods could be obtained for congruent
$\triangle OAB,\ \triangle O'A'B'$ with $\angle AOB=\pi$. 
In this case, $\bm{c}=L\bm{m}$, $V_3=0$ and $V_1=2L$. 
$$
V_2=|\bm{c}\times \bm{a'}|+|\bm{c}\times \bm{b'}|
+|\bm{a}\times \bm{c'}|+|\bm{b}\times \bm{c'}|
=|\bm{c}\times (\bm{a'}-\bm{b'})|
+|(\bm{a}-\bm{b})\times \bm{c'}|
=2|\bm{c}\times \bm{c'}|.
$$
Hence
$$
V=2L^2D|\bm{m}\times \bm{m'}|+2\pi LD^2+\frac{4}{3}\pi D^3
$$
which is a constant different from Onsager's form. 

\subsubsection{Quadratic projection of the excluded-volume potential}
The above derivation for excluded volume is valid for any pair of triangles. 
Now we suppose that $T$ is isoceles with top corner $\theta$ and length of 
lateral sides $L/2$. Two triangles are given by $T_1=PT$ and $T_2=P'T$. 
The unit vectors along the edges of two triangles are written as follows. 
\begin{eqnarray*}
  \bm{e}_a=\frac{\bm{a}}{|\bm{a}|}=\bm{m}_1\cos\frac{\theta}{2}+\bm{m}_2\sin\frac{\theta}{2},\ 
  \bm{e}_b=\frac{\bm{b}}{|\bm{b}|}=-\bm{m}_1\cos\frac{\theta}{2}+\bm{m}_2\sin\frac{\theta}{2},\ 
  \bm{e}_c=\frac{\bm{c}}{|\bm{c}|}=-\bm{m}_2,\\
  \bm{e'}_a=\frac{\bm{a'}}{|\bm{a'}|}=\bm{m'}_1\cos\frac{\theta}{2}+\bm{m'}_2\sin\frac{\theta}{2},\ 
  \bm{e'}_b=\frac{\bm{b'}}{|\bm{b'}|}=-\bm{m'}_1\cos\frac{\theta}{2}+\bm{m'}_2\sin\frac{\theta}{2},\ 
  \bm{e'}_c=\frac{\bm{c'}}{|\bm{c'}|}=-\bm{m'}_2
\end{eqnarray*}
with $|\bm{a}|=|\bm{b}|=|\bm{a'}|=|\bm{b'}|=L/2$ 
and $|\bm{c}|=|\bm{c'}|=L\sin\frac{\theta}{2}$. 
We aim to project $V$ onto the space spanned by
$$
Q=\{1, p_{11}, p_{11}^2, p_{12}^2,p_{21}^2, p_{22}^2\}. 
$$
Note that the following functions in span$\{Q\}$ are mutually orthogonal: 
$$
1,p_{11}, \frac{1}{2}(3p_{11}^2-1), \sqrt{3}(p_{12}^2+\frac{1}{2}(p_{11}^2-1)),
\sqrt{3}(p_{21}^2+\frac{1}{2}(p_{11}^2-1)), 
2p_{22}^2+(p_{12}^2+p_{21}^2)+\frac{1}{2}p_{11}^2-\frac{3}{2}. 
$$
We focus on the even-order terms first. Let
\begin{eqnarray*}
  k_0&=&\int \d\nu(\bar{P})V(\bar{P}),\\
  k_1&=&\int \d\nu(\bar{P})V(\bar{P})p_{11}^2,\\
  k_2&=&\int \d\nu(\bar{P})V(\bar{P})p_{12}^2=\int \d\nu(\bar{P})V(\bar{P})p_{21}^2,\\
  k_3&=&\int \d\nu(\bar{P})V(\bar{P})p_{22}^2. 
\end{eqnarray*}
The even-order part of projection will be written as
\begin{align*}
  5&\left[(\frac{3}{2}k_1-\frac{1}{2}k_0)(\frac{3}{2}p_{11}^2-\frac{1}{2})
    +3(k_2+\frac{1}{2}(k_1-k_0))(p_{12}^2+p_{21}^2+p_{11}^2-1)\right.\\
    &\left.+4(k_3+k_2+\frac{1}{4}k_1-\frac{3}{4}k_0)(p_{22}^2
    +\frac{1}{2}(p_{12}^2+p^2_{21})+\frac{1}{4}p_{11}^2-\frac{3}{4}) \right]. 
\end{align*}
By comparing the coefficients, we have
\begin{eqnarray}
  c_2&=&5(4k_1+4k_2+k_3-3k_0),\label{cc2}\\
  c_3&=&5(k_1+4k_2+4k_3-3k_0), \label{cc3}\\
  c_4&=&5(2k_1+5k_2+2k_3-3k_0).\label{cc4}
\end{eqnarray}

In the above, $k_0,k_1,k_2,k_3$ can be evaluated analytically. We use the 
notation $p_{ij}(\bar{P})$ to represent the $(i,j)$ element of $\bar{P}$.
First we point out that 
\begin{eqnarray}
&&\int_{SO_3}\d\nu(\bar{P})V_2(\bar{P})p^2_{ij}(\bar{P})\nonumber\\
&=&\int_{SO_3}\d\nu(\bar{P})p_{ij}^2(\bar{P})
\left(|\bm{a}\times\bm{b}|+|\bm{a'}\times\bm{b'}|+
\sum_{\bm{e}\in\{\bm{a,b,c}\},\bm{e'}\in\{\bm{a',b',c'}\}}
\frac{1}{2}|\bm{e}\times\bm{e'}|\right).\label{cent}
\end{eqnarray}
In fact, $V_2(\bar{P})=V_2(PT-P'T)$. 
Let $J=\mbox{diag}(-1,-1,1)$, then $JT=-T$. 
Thereby
$$V_2(\bar{P}J)=V_2(PT-P'JT)=V_2(PT+P'T).$$ 
By (\ref{PntRef}) we have
$$
V_2(\bar{P})+V_2(\bar{P}J)=
2\big(|\bm{a}\times\bm{b}|+|\bm{a'}\times\bm{b'}|\big)+
\sum_{\bm{e}\in\{\bm{a,b,c}\},\bm{e'}\in\{\bm{a',b',c'}\}}|\bm{e}\times\bm{e'}|. 
$$
Meanwhile $p_{ij}(\bar{P}J)=p_{ij}(\bar{P})$, therefore (\ref{cent}) holds. 
We need to calculate the terms like
\begin{equation}\label{cos}
\int_{SO_3}\d\nu(\bar{P})p_{ij}^2|\bm{a}\times\bm{b}\cdot\bm{a'}|
=\frac{1}{8}L^3\sin\theta\int_{SO_3}\d\nu(\bar{P})p_{ij}^2|\bm{m}_3\cdot
\bm{e'}_a|
\end{equation}
and
\begin{equation}\label{sin}
\int_{SO_3}\d\nu(\bar{P})p_{ij}^2|\bm{a}\times\bm{a'}|
=\frac{1}{4}L^2\int_{SO_3}\d\nu(\bar{P})p_{ij}^2
|\bm{e}_a\times\bm{e'}_a|. 
\end{equation}

We describe the strategy to compute integrals
$$
\int_{SO_3}\d\nu(\bar{P})p_{ij}^2|\bm{e}\times\bm{e'}|,\quad 
\int_{SO_3}\d\nu(\bar{P})p_{ij}^2|\bm{e}\cdot\bm{e'}|, 
$$
where $\bm{e},\bm{e'}$ are unit vectors. The following formula is needed. 
\begin{equation}
  \int_{SO_3}\d\nu(\bar{P})f(\bar{P})=\int_{SO_3}\d\nu(\bar{P})f(R_1^{-1}\bar{P}R_2), 
  \quad \forall R_1,R_2\in SO_3.
\end{equation}
Choose $R_1$ and $R_2$ such that
$$
R_1\bm{e}=\bm{m}_1,\quad R_2\bm{e'}=\bm{m'}_1. 
$$
The integral is rewritten as 
\begin{eqnarray*}
\int_{SO_3}\d\nu(\bar{P})p_{ij}^2|\bm{e}\times\bm{e'}|&=&
\int_{SO_3}\d\nu(\bar{P})p^2_{ij}(R_1^{-1}\bar{P}R_2)|\bm{m}_1\times\bm{m'}_1|\\
&=&\int^{\pi}_0 \d\alpha\int^{2\pi}_0\d\beta\int^{2\pi}_0\d\gamma
\frac{\sin\alpha}{8\pi^2} |\sin\alpha|Q(\alpha,\beta,\gamma), 
\end{eqnarray*}
in which $Q$ is a trigonometric polynomial of $\alpha,\beta,\gamma$. 
When the cross product is replaced by dot product, $|\sin\alpha|$ is 
substituted with $|\cos\alpha|$. 
We compute (\ref{cos}) as an example. Define $R_1$ and $R_2$ by
\begin{align*}
  R_1\bm{m}_1&=-\bm{m}_3,&R_1\bm{m}_2&=\bm{m}_2,&R_1\bm{m}_3&=\bm{m}_1,\\
  R_2\bm{m'}_1&=\bm{m'}_1\cos\frac{\theta}{2}-\bm{m'}_2\sin\frac{\theta}{2},&
  R_2\bm{m'}_2&=\bm{m'}_2\cos\frac{\theta}{2}+\bm{m'}_1\sin\frac{\theta}{2},&
  R_2\bm{m'}_3&=\bm{m'}_3.
\end{align*}
Then we have
$$
p_{11}(R_1^{-1}\bar{P}R_2)=-p_{31}\cos\frac{\theta}{2}+p_{32}\sin\frac{\theta}{2}.
$$
Hence 
\begin{eqnarray*}
&&\int_{SO_3}\d\nu(\bar{P})p_{11}^2|\bm{m}_3\cdot
(\bm{m'}_1\cos\frac{\theta}{2}+\bm{m'}_2\sin\frac{\theta}{2})|\\
&=&\int_{SO_3}\d\nu(\bar{P})(-p_{31}\cos\frac{\theta}{2}+p_{32}
\sin\frac{\theta}{2})^2|\bm{m}_1\cdot\bm{m'}_1|\\
&=&\int_0^{\pi}\d\alpha\int_0^{2\pi}\d\beta\int_{0}^{2\pi}
\d\gamma\frac{\sin\alpha}{8\pi^2}|\cos\alpha|\\
&&\big(-\sin\alpha\sin\beta\cos\frac{\theta}{2}+\sin\frac{\theta}{2}
(\cos\alpha\sin\beta\cos\gamma+\cos\beta\sin\gamma)\big)^2\\
&=&\frac{1}{8}\cos^2\frac{\theta}{2}+\frac{3}{16}\sin^2\frac{\theta}{2}.
\end{eqnarray*}
The other terms could be handled similarly. All the results are listed in 
Table.\ref{projcoe} at the end of the article. 
By collecting those results, we get
\begin{eqnarray*}
  2k_0&=&\frac{1}{4}cL^3\sin\theta(1+\sin\frac{\theta}{2})
  +\frac{\pi}{4}cL^2D(1+2\sin\frac{\theta}{2}+\sin^2\frac{\theta}{2})+3C,\\
  2k_1&=&\frac{1}{32}cL^3\sin\theta(2\cos^2\frac{\theta}{2}
  +3\sin^2\frac{\theta}{2}+3\sin\frac{\theta}{2})
  +\frac{\pi}{64}cL^2D\Big[4\cos^4\frac{\theta}{2}+5\sin^4\frac{\theta}{2}\\
  &&+12\sin^2\frac{\theta}{2}\cos^2\frac{\theta}{2}+2\sin\frac{\theta}{2}
  (6\cos^2\frac{\theta}{2}+5\sin^2\frac{\theta}{2})+5\sin^2\frac{\theta}{2}\Big]+C,\\
  2k_2&=&\frac{1}{64}cL^3\sin\theta(5+5\sin\frac{\theta}{2})
  +\frac{\pi}{64}cL^2D\Big[6\cos^4\frac{\theta}{2}+6\sin^4\frac{\theta}{2}\\
  &&+9\sin^2\frac{\theta}{2}\cos^2\frac{\theta}{2}+\sin\frac{\theta}{2}
  (9\cos^2\frac{\theta}{2}+12\sin^2\frac{\theta}{2})+6\sin^2\frac{\theta}{2}\Big]+C,\\
  2k_3&=&\frac{1}{32}cL^3\sin\theta(3\cos^2\frac{\theta}{2}
  +2\sin^2\frac{\theta}{2}+2\sin\frac{\theta}{2})
  +\frac{\pi}{64}cL^2D\Big[5\cos^4\frac{\theta}{2}+4\sin^4\frac{\theta}{2}\\
  &&+12\sin^2\frac{\theta}{2}\cos^2\frac{\theta}{2}+2\sin\frac{\theta}{2}
  (6\cos^2\frac{\theta}{2}+4\sin^2\frac{\theta}{2})+4\sin^2\frac{\theta}{2}\Big]+C, 
\end{eqnarray*}
where $C$ is a constant
$$
3C=\frac{1}{4}L^2D\sin\theta+D^2L(1+\sin\frac{\theta}{2})+\frac{4}{3}\pi D^3.
$$
By (\ref{cc2})-(\ref{cc4}), we get (\ref{c_2})-(\ref{c_4}). 

The computation of $c_1$ is complicated. Note that $V_3$ does not contribute to 
$c_1$. In fact, it is obvious that $V_3(\bar{P}J)=V_3(\bar{P})$ and 
$p_{11}(\bar{P}J)=-p_{11}(\bar{P})$, which yield
$$
\int_{SO_3}\d\nu(\bar{P})V_3p_{11}=0.
$$
Therefore
$$
c_1=\frac{3}{8}cDL^2K(\theta). 
$$
Denote
\begin{align*}
  I_{aa}&=|\bm{e}_a\times\bm{e'}_a|, &
  I_{ab}&=|\bm{e}_a\times\bm{e'}_b|, &
  I_{ac}&=2\sin\frac{\theta}{2}|\bm{e}_a\times\bm{e'}_c|, \\
  I_{ba}&=|\bm{e}_b\times\bm{e'}_a|, &
  I_{bb}&=|\bm{e}_b\times\bm{e'}_b|, &
  I_{bc}&=2\sin\frac{\theta}{2}|\bm{e}_b\times\bm{e'}_c|, \\
  I_{ca}&=2\sin\frac{\theta}{2}|\bm{e}_c\times\bm{e'}_a|, &
  I_{cb}&=2\sin\frac{\theta}{2}|\bm{e}_c\times\bm{e'}_b|, &
  I_{cc}&=4\sin^2\frac{\theta}{2}|\bm{e}_c\times\bm{e'}_c|.
\end{align*}
$K(\theta)$ is written as
\begin{equation}
  \begin{split}
  K(\theta)&=\int_{SO_3}\d\nu(\bar{P})p_{11}\\
&\left[(I_{aa}+I_{ab}+I_{ba}+I_{bb}+I_{cc})
 \chi_{\{ (\bm{m}_3\cdot\bm{a'})(\bm{m}_3\cdot\bm{b'})> 0,
  (\bm{m'}_3\cdot\bm{a})(\bm{m'}_3\cdot\bm{b})> 0,
 (\bm{m}_3\cdot\bm{c'})(\bm{m'}_3\cdot\bm{c})<0\}}\right.
\\
&(I_{ac}+I_{bc}+I_{ca}+I_{cb})
 \chi_{\{ (\bm{m}_3\cdot\bm{a'})(\bm{m}_3\cdot\bm{b'})> 0,
  (\bm{m'}_3\cdot\bm{a})(\bm{m'}_3\cdot\bm{b})> 0,
 (\bm{m}_3\cdot\bm{c'})(\bm{m'}_3\cdot\bm{c})>0\}}
\\
&(I_{ab}+I_{ac}+I_{bb}+I_{bc}+I_{ca})
 \chi_{\{ (\bm{m}_3\cdot\bm{b'})(\bm{m}_3\cdot\bm{c'})> 0,
  (\bm{m'}_3\cdot\bm{a})(\bm{m'}_3\cdot\bm{b})> 0,
 (\bm{m}_3\cdot\bm{a'})(\bm{m'}_3\cdot\bm{c})<0\}}
\\
&(I_{aa}+I_{ba}+I_{cb}+I_{cc})
 \chi_{\{ (\bm{m}_3\cdot\bm{b'})(\bm{m}_3\cdot\bm{c'})> 0,
  (\bm{m'}_3\cdot\bm{a})(\bm{m'}_3\cdot\bm{b})> 0,
 (\bm{m}_3\cdot\bm{a'})(\bm{m'}_3\cdot\bm{c})>0\}}
\\
&(I_{ac}+I_{aa}+I_{bc}+I_{ba}+I_{cb})
 \chi_{\{ (\bm{m}_3\cdot\bm{c'})(\bm{m}_3\cdot\bm{a'})> 0,
  (\bm{m'}_3\cdot\bm{a})(\bm{m'}_3\cdot\bm{b})> 0,
 (\bm{m}_3\cdot\bm{b'})(\bm{m'}_3\cdot\bm{c})<0\}}
\\
&(I_{ab}+I_{bb}+I_{cc}+I_{ca})
 \chi_{\{ (\bm{m}_3\cdot\bm{c'})(\bm{m}_3\cdot\bm{a'})> 0,
  (\bm{m'}_3\cdot\bm{a})(\bm{m'}_3\cdot\bm{b})> 0,
 (\bm{m}_3\cdot\bm{b'})(\bm{m'}_3\cdot\bm{c})>0\}}
\\
&(I_{ba}+I_{bb}+I_{ca}+I_{cb}+I_{ac})
 \chi_{\{ (\bm{m}_3\cdot\bm{a'})(\bm{m}_3\cdot\bm{b'})> 0,
  (\bm{m'}_3\cdot\bm{b})(\bm{m'}_3\cdot\bm{c})> 0,
 (\bm{m}_3\cdot\bm{c'})(\bm{m'}_3\cdot\bm{a})<0\}}
\\
&(I_{bc}+I_{cc}+I_{aa}+I_{ab})
 \chi_{\{ (\bm{m}_3\cdot\bm{a'})(\bm{m}_3\cdot\bm{b'})> 0,
  (\bm{m'}_3\cdot\bm{b})(\bm{m'}_3\cdot\bm{c})> 0,
 (\bm{m}_3\cdot\bm{c'})(\bm{m'}_3\cdot\bm{a})>0\}}
\\
&(I_{bb}+I_{bc}+I_{cb}+I_{cc}+I_{aa})
 \chi_{\{ (\bm{m}_3\cdot\bm{b'})(\bm{m}_3\cdot\bm{c'})> 0,
  (\bm{m'}_3\cdot\bm{b})(\bm{m'}_3\cdot\bm{c})> 0,
 (\bm{m}_3\cdot\bm{a'})(\bm{m'}_3\cdot\bm{a})<0\}}
\\
&(I_{ba}+I_{ca}+I_{ab}+I_{ac})
 \chi_{\{ (\bm{m}_3\cdot\bm{b'})(\bm{m}_3\cdot\bm{c'})> 0,
  (\bm{m'}_3\cdot\bm{b})(\bm{m'}_3\cdot\bm{c})> 0,
 (\bm{m}_3\cdot\bm{a'})(\bm{m'}_3\cdot\bm{a})>0\}}
\\
&(I_{bc}+I_{ba}+I_{cc}+I_{ca}+I_{ab})
 \chi_{\{ (\bm{m}_3\cdot\bm{c'})(\bm{m}_3\cdot\bm{a'})> 0,
  (\bm{m'}_3\cdot\bm{b})(\bm{m'}_3\cdot\bm{c})> 0,
 (\bm{m}_3\cdot\bm{b'})(\bm{m'}_3\cdot\bm{a})<0\}}
\\
&(I_{bb}+I_{cb}+I_{ac}+I_{aa})
 \chi_{\{ (\bm{m}_3\cdot\bm{c'})(\bm{m}_3\cdot\bm{a'})> 0,
  (\bm{m'}_3\cdot\bm{b})(\bm{m'}_3\cdot\bm{c})> 0,
 (\bm{m}_3\cdot\bm{b'})(\bm{m'}_3\cdot\bm{a})>0\}}
\\
&(I_{ca}+I_{cb}+I_{aa}+I_{ab}+I_{bc})
 \chi_{\{ (\bm{m}_3\cdot\bm{a'})(\bm{m}_3\cdot\bm{b'})> 0,
  (\bm{m'}_3\cdot\bm{c})(\bm{m'}_3\cdot\bm{a})> 0,
 (\bm{m}_3\cdot\bm{c'})(\bm{m'}_3\cdot\bm{b})<0\}}
\\
&(I_{cc}+I_{ac}+I_{ba}+I_{bb})
 \chi_{\{ (\bm{m}_3\cdot\bm{a'})(\bm{m}_3\cdot\bm{b'})> 0,
  (\bm{m'}_3\cdot\bm{c})(\bm{m'}_3\cdot\bm{a})> 0,
 (\bm{m}_3\cdot\bm{c'})(\bm{m'}_3\cdot\bm{b})>0\}}
\\
&(I_{cb}+I_{cc}+I_{ab}+I_{ac}+I_{ba})
 \chi_{\{ (\bm{m}_3\cdot\bm{b'})(\bm{m}_3\cdot\bm{c'})> 0,
  (\bm{m'}_3\cdot\bm{c})(\bm{m'}_3\cdot\bm{a})> 0,
 (\bm{m}_3\cdot\bm{a'})(\bm{m'}_3\cdot\bm{b})<0\}}
\\
&(I_{ca}+I_{aa}+I_{bb}+I_{bc})
 \chi_{\{ (\bm{m}_3\cdot\bm{b'})(\bm{m}_3\cdot\bm{c'})> 0,
  (\bm{m'}_3\cdot\bm{c})(\bm{m'}_3\cdot\bm{a})> 0,
 (\bm{m}_3\cdot\bm{a'})(\bm{m'}_3\cdot\bm{b})>0\}}
\\
&(I_{cc}+I_{ca}+I_{ac}+I_{aa}+I_{bb})
 \chi_{\{ (\bm{m}_3\cdot\bm{c'})(\bm{m}_3\cdot\bm{a'})> 0,
  (\bm{m'}_3\cdot\bm{c})(\bm{m'}_3\cdot\bm{a})> 0,
 (\bm{m}_3\cdot\bm{b'})(\bm{m'}_3\cdot\bm{b})<0\}}
\\
&\left.(I_{cb}+I_{ab}+I_{bc}+I_{ba})
 \chi_{\{ (\bm{m}_3\cdot\bm{c'})(\bm{m}_3\cdot\bm{a'})> 0,
  (\bm{m'}_3\cdot\bm{c})(\bm{m'}_3\cdot\bm{a})> 0,
 (\bm{m}_3\cdot\bm{b'})(\bm{m'}_3\cdot\bm{b})>0\}}\right]. 
  \end{split}
\label{Ktheta}
\end{equation}

\subsection{The excluded-volume potential of bent-core molecules}
Denote by $\bm{p}_1$ and $\bm{p}_2$ two unit vectors along the arms of molecule.
The excluded region of two molecules is the union of four spheroparallelograms 
$$
V_{ij}=O\bm{p}_i\bm{p'}_j+B_D,\qquad i,j=1,2.
$$ 
Thus the excluded volume can be written as 
$$
V=\sum|V_{ij}|-\sum|V_{ij}\cap V_{i'j'}|+\sum|V_{ij}\cap V_{i'j'}\cap V_{i''j''}|-|V_{11}\cap V_{12}\cap V_{21}\cap V_{22}|. 
$$
We have already known that 
$$
|V_{ij}|=2L^2D|\bm{p}_i\times\bm{p'}_j|+2\pi LD^2+\frac{4}{3}\pi D^3.
$$ 
So we only need to compute the volumes of the intersections above. 

When calculating the volume of a region $U$, we can write
$$
|U|=\int\d x\d y m(\Omega(x,y))
$$
where $m(\cdot)$ denotes the measure of a set and 
$\Omega(x,y)=\{z|(x,y,z)\in U\}$. 
Because $V_{ij}$ is convex, $\Omega(x,y)$ is an interval 
$[l_{ij}(x,y)$, $u_{ij}(x,y)]$ for $U=V_{ij}$. Thus
\begin{eqnarray*}
|V_{ij}\cap V_{i'j'}|&=&\int\d x\d y \big[\min\{u_{ij},u_{i'j'}\}-\max\{l_{ij},l_{i'j'}\}\big]^+,\\
|V_{ij}\cap V_{i'j'}\cap V_{i''j''}|&=&\int\d x\d y \big[\min\{u_{ij},u_{i'j'},u_{i''j''}\}-\max\{l_{ij},l_{i'j'},l_{i''j''}\}\big]^+,\\
|V_{11}\cap V_{12}\cap V_{21}\cap V_{22}|&=&\int\d x\d y 
\big[\min\{u_{11},u_{12},u_{21},u_{22}\}-\max\{l_{11},l_{12},l_{21},l_{22}\}\big]^+,
\end{eqnarray*}
where $x^+=\max\{x,0\}$. 
Now the problem turns into computing $l_{ij}(x,y)$ and $u_{ij}(x,y)$. 

Put one molecule in the plane $xOy$ with the arrowhead at $O$ and $\bm{m}_1$ 
along $-x$. Then $\bm{p}_{1,2}=L(\cos\frac{\theta}{2},\pm\sin\frac{\theta}{2},0)$.
We describe how to compute $u(x,y)$ of the spheroparallelogram $O\bm{p}_1\bm{
p'}_1+B_D$, where $\bm{p'}_1/L=(p,q,r)$ is a unit vector. A spheroparallelogram 
consists of a parallelpiped, four half cylinders at each edge of parallelogram, 
and four corners, each of which is enclosed by two planes and a sphere. 
Classify $u(x,y)$ into three cases by where $(x,y,u(x,y))$ lies: the 
parallelpiped; one of the four half cylinders; one of the four spheres. 
For the first case, the distance of $(x,y,u(x,y))$ to plane 
$O\bm{p}_1\bm{p'}_1$ equals to $D$. The normal vector of $O\bm{p}_1\bm{p'}_1$ is
$$
\bm{p}_1\times\bm{p'}_1=(A,B,C)=(r\sin\frac{\theta}{2},-r\cos\frac{\theta}{2}, 
q\cos\frac{\theta}{2}-p\sin\frac{\theta}{2}).
$$
Hence 
$$
u(x,y)=\frac{D\sqrt{A^2+B^2+C^2}}{|C|}-\frac{Ax+By}{C}. 
$$
For the second case, the distance equals to $D$ bewteen $(x,y,u(x,y))$ and the 
axis of one of the four half cylinders. Thereby $u(x,y)$ is the larger root of 
$$
(x-x_0)^2+(y-y_0)^2+\big(u(x,y)-z_0\big)^2-\big[a(x-x_0)+b(y-y_0)+c(u(x,y)-z_0)\big]^2=D^2.
$$
In the above, $(x_0,y_0,z_0)$ is any point on the axis, which may take $O$ or
$\Big(L(\cos\frac{\theta}{2}+p),L(\sin\frac{\theta}{2}+q),Lr\Big)$; 
$(a,b,c)$ is the unit vector along the axis, which may take
$(\cos\frac{\theta}{2},\sin\frac{\theta}{2},0)$ or $(p,q,r)$. 
For the third case, $u(x,y)$ is the larger root of 
$$
(x-x_0)^2+(y-y_0)^2+(z-z_0)^2=D^2
$$
where $(x_0,y_0,z_0)$ is one of the four vertices of the parallelogram 
$O\bm{p}_1\bm{p'}_1$. 

\begin{figure}
\centering
\includegraphics[width=0.55\textwidth,keepaspectratio]{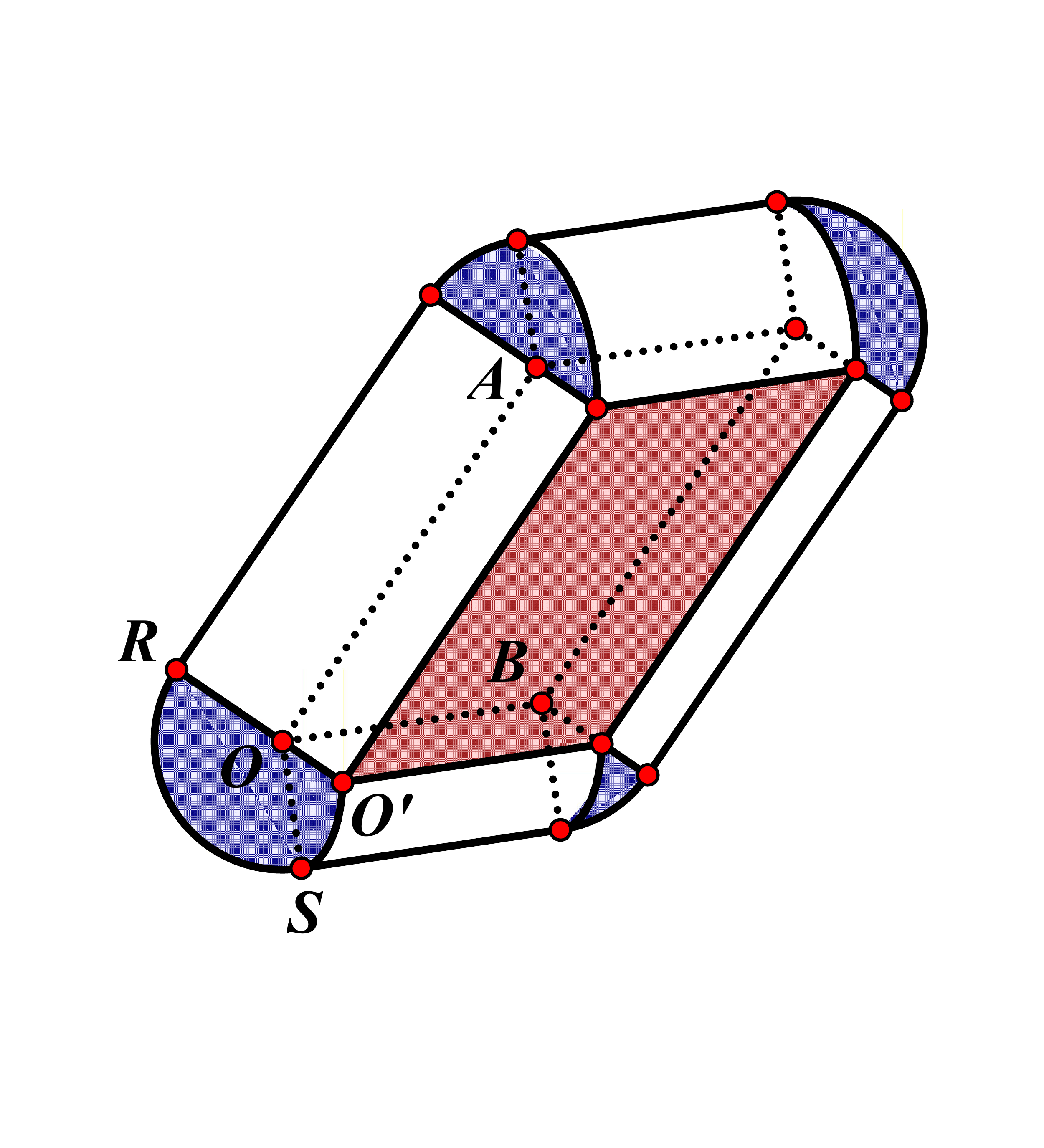}
\caption{Integration regions, divided into three cases. }\label{region}
\end{figure}
All remaining is to clarify the region of three cases. In Fig.\ref{region}, 
they are coloured by red, white and blue respectively. The region contains 
all points whose distance to the central parallelogram (drawn in 
dotted line in Fig.\ref{region}), which is spanned by 
$\overrightarrow{OA}=(L\cos\frac{\theta}{2},L\sin\frac{\theta}{2})$ 
and $\overrightarrow{OB}=(Lp,Lq)$, is no more than $D$. 
It consists of the central parallelogram, four rectangles at each edge and 
four sectors at each corner. The red region is the projection to plane $xOy$ 
of the parallelogram $O_1\bm{p}_1\bm{p'}_1$ with 
$$
\overrightarrow{OO_{1}}=D\mbox{sgn}C\frac{\bm{p}_1\times\bm{p'}_1}
{|\bm{p}_1\times\bm{p'}_1|}.
$$
It is obtained by shifting the central 
parallelogram along $\overrightarrow{OO'}$, where $O'$ locates at 
$$
\frac{D\mbox{sgn}(q\cos\frac{\theta}{2}-p\sin\frac{\theta}{2})}
{\sqrt{1-(p\cos\frac{\theta}{2}+q\sin\frac{\theta}{2})^2}}
(r\sin\frac{\theta}{2},-r\cos\frac{\theta}{2}). 
$$
Two of the four white regions are rectangles. 
The other two white regions are enclosed by two line segments and two 
elliptical arcs. Each elliptical arc connects a vertex of the red region and 
a vertex of the rectangles on the boundary, such as $O'S$. It is the projection 
of a curve to $xOy$. The curve is part of the intersection of the sphere
$$
(x-x_0)^2+(y-y_0)^2+(z-z_0)^2=D^2
$$
and the plane
$$
p(x-x_0)+q(y-y_0)+r(z-z_0)=0
$$
where $(x_0,y_0,z_0)$ is one of the vertices of $O\bm{p}_1\bm{p'}_1$. 
By eliminating $z$, we get the equation of the curve: 
$$
\big[p(x-x_0)+q(y-y_0)\big]^2+r^2\big[(x-x_0)^2+(y-y_0)^2\big]=r^2D^2. 
$$
The blue regions are those enclosed by a line segment, an elliptical arc 
defined above, and a circular arc on the boundary. 

\textbf{Acknowledgements.} 
P. Zhang is partly supported by NSF of China under Grant 50930003 and 21274005.

\renewcommand\arraystretch{2.0}
\begin{landscape}
\begin{table}
  \caption{Coefficients: $\int_{SO_3}\d\nu fp^2_{ij}$}\tiny
  \begin{tabular}{|c||c|c|c|c|}
    \hline
    Function & $p^2_{11}$ & $p_{22}^2$ & $p_{12}^2$ & $p_{21}^2$ \\\hline

    $|\bm{m'}_3\cdot \bm{e}_a|$&
    $\frac{1}{8}\cos^2\frac{\theta}{2}+\frac{3}{16}\sin^2\frac{\theta}{2}$ &
    $\frac{3}{16}\cos^2\frac{\theta}{2}+\frac{1}{8}\sin^2\frac{\theta}{2}$ &
    $\frac{1}{8}\cos^2\frac{\theta}{2}+\frac{3}{16}\sin^2\frac{\theta}{2}$ &
    $\frac{3}{16}\cos^2\frac{\theta}{2}+\frac{1}{8}\sin^2\frac{\theta}{2}$
    \\\hline

    $|\bm{m'}_3\cdot \bm{e}_b|$&
    $\frac{1}{8}\cos^2\frac{\theta}{2}+\frac{3}{16}\sin^2\frac{\theta}{2}$ &
    $\frac{3}{16}\cos^2\frac{\theta}{2}+\frac{1}{8}\sin^2\frac{\theta}{2}$ &
    $\frac{1}{8}\cos^2\frac{\theta}{2}+\frac{3}{16}\sin^2\frac{\theta}{2}$ &
    $\frac{3}{16}\cos^2\frac{\theta}{2}+\frac{1}{8}\sin^2\frac{\theta}{2}$
    \\\hline

    $|\bm{m'}_3\cdot \bm{e}_c|$ & $\frac{3}{16}$ & $\frac{1}{8}$ & 
    $\frac{3}{16}$ & $\frac{1}{8}$
    \\\hline

    $|\bm{m}_3\cdot \bm{e'}_a|$&
    $\frac{1}{8}\cos^2\frac{\theta}{2}+\frac{3}{16}\sin^2\frac{\theta}{2}$ &
    $\frac{3}{16}\cos^2\frac{\theta}{2}+\frac{1}{8}\sin^2\frac{\theta}{2}$ &
    $\frac{3}{16}\cos^2\frac{\theta}{2}+\frac{1}{8}\sin^2\frac{\theta}{2}$ &
    $\frac{1}{8}\cos^2\frac{\theta}{2}+\frac{3}{16}\sin^2\frac{\theta}{2}$ 
    \\\hline

    $|\bm{m}_3\cdot \bm{e'}_b|$&
    $\frac{1}{8}\cos^2\frac{\theta}{2}+\frac{3}{16}\sin^2\frac{\theta}{2}$ &
    $\frac{3}{16}\cos^2\frac{\theta}{2}+\frac{1}{8}\sin^2\frac{\theta}{2}$ &
    $\frac{3}{16}\cos^2\frac{\theta}{2}+\frac{1}{8}\sin^2\frac{\theta}{2}$ &
    $\frac{1}{8}\cos^2\frac{\theta}{2}+\frac{3}{16}\sin^2\frac{\theta}{2}$ 
    \\\hline

    $|\bm{m}_3\cdot \bm{e'}_c|$ & $\frac{3}{16}$ & $\frac{1}{8}$ & 
    $\frac{1}{8}$ & $\frac{3}{16}$
    \\\hline

    $|\bm{e}_a\times \bm{e'}_a|$&
    $\pi\left(\frac{1}{16}\cos^4\frac{\theta}{2}+\frac{5}{64}\sin^4\frac{\theta}
    {2}+\frac{3}{16}\sin^2\frac{\theta}{2}\cos^2\frac{\theta}{2}\right)$ &
    $\pi\left(\frac{1}{16}\cos^4\frac{\theta}{2}+\frac{5}{64}\sin^4\frac{\theta}
    {2}+\frac{3}{16}\sin^2\frac{\theta}{2}\cos^2\frac{\theta}{2}\right)$ &
    $\pi\left(\frac{1}{16}\cos^4\frac{\theta}{2}+\frac{5}{64}\sin^4\frac{\theta}
    {2}+\frac{3}{16}\sin^2\frac{\theta}{2}\cos^2\frac{\theta}{2}\right)$ &
    $\pi\left(\frac{1}{16}\cos^4\frac{\theta}{2}+\frac{5}{64}\sin^4\frac{\theta}
    {2}+\frac{3}{16}\sin^2\frac{\theta}{2}\cos^2\frac{\theta}{2}\right)$\\\hline

    $|\bm{e}_a\times \bm{e'}_b|$&
    $\pi\left(\frac{1}{16}\cos^4\frac{\theta}{2}+\frac{5}{64}\sin^4\frac{\theta}
    {2}+\frac{3}{16}\sin^2\frac{\theta}{2}\cos^2\frac{\theta}{2}\right)$ &
    $\pi\left(\frac{1}{16}\cos^4\frac{\theta}{2}+\frac{5}{64}\sin^4\frac{\theta}
    {2}+\frac{3}{16}\sin^2\frac{\theta}{2}\cos^2\frac{\theta}{2}\right)$ &
    $\pi\left(\frac{1}{16}\cos^4\frac{\theta}{2}+\frac{5}{64}\sin^4\frac{\theta}
    {2}+\frac{3}{16}\sin^2\frac{\theta}{2}\cos^2\frac{\theta}{2}\right)$ &
    $\pi\left(\frac{1}{16}\cos^4\frac{\theta}{2}+\frac{5}{64}\sin^4\frac{\theta}
    {2}+\frac{3}{16}\sin^2\frac{\theta}{2}\cos^2\frac{\theta}{2}\right)$\\\hline

    $|\bm{e}_a\cdot \bm{e'}_c|$&
    $\pi\left(\frac{3}{32}\cos^2\frac{\theta}{2}+\frac{5}{64}\sin^2\frac{\theta}{2}\right)$ &
    $\pi\left(\frac{3}{32}\cos^2\frac{\theta}{2}+\frac{1}{16}\sin^2\frac{\theta}{2}\right)$ &
    $\pi\left(\frac{1}{16}\cos^2\frac{\theta}{2}+\frac{3}{32}\sin^2\frac{\theta}{2}\right)$ &
    $\pi\left(\frac{5}{64}\cos^2\frac{\theta}{2}+\frac{3}{32}\sin^2\frac{\theta}{2}\right)$ 
    \\\hline

    $|\bm{e}_b\times \bm{e'}_a|$&
    $\pi\left(\frac{1}{16}\cos^4\frac{\theta}{2}+\frac{5}{64}\sin^4\frac{\theta}
    {2}+\frac{3}{16}\sin^2\frac{\theta}{2}\cos^2\frac{\theta}{2}\right)$ &
    $\pi\left(\frac{1}{16}\cos^4\frac{\theta}{2}+\frac{5}{64}\sin^4\frac{\theta}
    {2}+\frac{3}{16}\sin^2\frac{\theta}{2}\cos^2\frac{\theta}{2}\right)$ &
    $\pi\left(\frac{1}{16}\cos^4\frac{\theta}{2}+\frac{5}{64}\sin^4\frac{\theta}
    {2}+\frac{3}{16}\sin^2\frac{\theta}{2}\cos^2\frac{\theta}{2}\right)$ &
    $\pi\left(\frac{1}{16}\cos^4\frac{\theta}{2}+\frac{5}{64}\sin^4\frac{\theta}
    {2}+\frac{3}{16}\sin^2\frac{\theta}{2}\cos^2\frac{\theta}{2}\right)$\\\hline

    $|\bm{e}_b\times \bm{e'}_b|$&
    $\pi\left(\frac{1}{16}\cos^4\frac{\theta}{2}+\frac{5}{64}\sin^4\frac{\theta}
    {2}+\frac{3}{16}\sin^2\frac{\theta}{2}\cos^2\frac{\theta}{2}\right)$ &
    $\pi\left(\frac{1}{16}\cos^4\frac{\theta}{2}+\frac{5}{64}\sin^4\frac{\theta}
    {2}+\frac{3}{16}\sin^2\frac{\theta}{2}\cos^2\frac{\theta}{2}\right)$ &
    $\pi\left(\frac{1}{16}\cos^4\frac{\theta}{2}+\frac{5}{64}\sin^4\frac{\theta}
    {2}+\frac{3}{16}\sin^2\frac{\theta}{2}\cos^2\frac{\theta}{2}\right)$ &
    $\pi\left(\frac{1}{16}\cos^4\frac{\theta}{2}+\frac{5}{64}\sin^4\frac{\theta}
    {2}+\frac{3}{16}\sin^2\frac{\theta}{2}\cos^2\frac{\theta}{2}\right)$\\\hline

    $|\bm{e}_b\cdot \bm{e'}_c|$&
    $\pi\left(\frac{3}{32}\cos^2\frac{\theta}{2}+\frac{5}{64}\sin^2\frac{\theta}{2}\right)$ &
    $\pi\left(\frac{3}{32}\cos^2\frac{\theta}{2}+\frac{1}{16}\sin^2\frac{\theta}{2}\right)$ &
    $\pi\left(\frac{1}{16}\cos^2\frac{\theta}{2}+\frac{3}{32}\sin^2\frac{\theta}{2}\right)$ &
    $\pi\left(\frac{5}{64}\cos^2\frac{\theta}{2}+\frac{3}{32}\sin^2\frac{\theta}{2}\right)$ 
    \\\hline

    $|\bm{e}_c\cdot \bm{e'}_a|$&
    $\pi\left(\frac{3}{32}\cos^2\frac{\theta}{2}+\frac{5}{64}\sin^2\frac{\theta}{2}\right)$ &
    $\pi\left(\frac{3}{32}\cos^2\frac{\theta}{2}+\frac{1}{16}\sin^2\frac{\theta}{2}\right)$ &
    $\pi\left(\frac{5}{64}\cos^2\frac{\theta}{2}+\frac{3}{32}\sin^2\frac{\theta}{2}\right)$ &
    $\pi\left(\frac{1}{16}\cos^2\frac{\theta}{2}+\frac{3}{32}\sin^2\frac{\theta}{2}\right)$ 
    \\\hline

    $|\bm{e}_c\cdot \bm{e'}_c|$ & $\frac{5}{64}\pi$ & $\frac{1}{16}\pi$ & 
    $\frac{3}{32}\pi$ & $\frac{3}{32}\pi$
    \\\hline
  \end{tabular}
\label{projcoe}
\end{table}
\end{landscape}


\begin{thebibliography}{99}
\bibitem{Largec} N. F. Carnahan and K. E. Starling, {\it Equation of State for Nonattracting Rigid Spheres}, J. Chem. Phys, 51, 635(1969).
\bibitem{LiqCryst} P. G. de Gennes and J. Prost, {\it The Physics of Liquid Crystals}, Clarendon Press, 1993. 
\bibitem{AxiSym2} I. Fatkullin and V. Slastikov, 
{\it Critical points of the Onsager functional on a sphere}, Nonlinearity, 2005.
\bibitem{Dipol} G. Ji, Q. Wang, P. Zhang and H. Zhou, 
{\it Study of phase transition in homogeneous, rigid extended nematics and magnetic suspensions using an order-reduction method}, 
Phys. Fluid, 18, 123103(2006). 
\bibitem{NT} E. H. Kim, O. N. Kadkin, So. Y. Kim and M. G. Choi, 
{\it Tetrahedratic Mesophases, Ambidextrous Chiral Domains and Helical Superstructures Produced by Achiral 1,1'-Disubstituted Ferrocene Derivatives}, 
Eur. J. Inorg. Chem, 2933(2011).
\bibitem{AxiSymMS} H. Liu, H. Zhang and P. Zhang, 
{\it Axial symmetry and classification of stationary solutions of Doi-Onsager 
equation on the sphere with Maier-Saupe potential}, 
Comm. Math. Sci, 3, 201(2005)
\bibitem{OrdPar} T. C. Lubensky and L. Radzihovsky, 
{\it Theory of bent-core liquid-crystal phases and phase transitions}, 
Phys. Rev. E, 66, 031704(2002). 
\bibitem{M_S} W. Maier and A. Z. Saupe, 
{\it Eine einfache molekulare theories des nametischen kristallinfl\"ussigen Zustandes}, 
Naturforsch, A13, 564(1958). 
\bibitem{BiLand} G. D. Matteis, A. M. Sonnet and E. G. Virga, 
{\it Landau theory for biaxial nematic liquid crystals with two order parameter tensors}, 
Continuum Mech. Thermodyn, 20, 347(2008). 
\bibitem{Mayer} J. E. Mayer and M. G. Mayer, {\it Statistical Mechanics}, Wiley, New York(1940).
\bibitem{Mulder} B. M. Mulder, {\it The excluded volume of hard sphero-zonotopes}, 
Mol. Phys. 103, 1411(2005). 
\bibitem{Ons} L. Onsager, {\it The effects of shape on the interaction of colloidal particles}, Ann. N. Y. Acad. Sci. 51, 627(1949). 
\bibitem{quadproj} R. Rosso and E. G. Virga, 
{\it Quadrupolar projection of excluded-volume interactions in biaxial nematic liquid crystals}, Phys. Rev. E, 72, 021712(2006). 
\bibitem{Convex} R. Schneider, {\it Convex bodies: the Brunn-Minkowski Theory}, 
Encyclopedia of Mathematics and its Applications Vol. 44, Cambridge University 
Press, Cambridge, U.K.(1993). 
\bibitem{Bi1} J. P. Starley, {\it Ordered phases of a liquid of biaxial particles}, Phys. Rev. A, 10, 1881(1974). 
\bibitem{JJAP} H. Takezoe and Y. Takanishi, 
{\it Bent-Core Liquid Crystals: Their Mysterious and Attractive World}, 
Jpn. J. Appl. Phys, 45, 597(2006). 
\bibitem{AxiSym3} H. Zhou, H. Wang, M. G. Forest and Q. Wang, 
{\it A new proof on axisymmetric equilibria of a three-dimensional Smoluchowski equation}, Nonlinearity, 2005.
\end{thebibliography}
\end{document}